\documentclass[10pt,final]{IEEEtran}

\ifCLASSINFOpdf
   \usepackage[pdftex]{graphicx}
\else
  \usepackage[dvips]{graphicx}
\fi
\usepackage{amssymb}
\usepackage[cmex10]{amsmath}
\usepackage{bm}
\usepackage{mathabx}
\usepackage{nccmath}

\usepackage{cite} 
\usepackage{float}

\usepackage[caption=false,font=footnotesize]{subfig}
\usepackage{setspace}
\usepackage{algorithm}
\usepackage{algpseudocode}

\newtheorem{theorem}{Theorem}[section]
\newtheorem{definition}[theorem]{Definition}

\newtheorem{corollary}[theorem]{Corollary}
\newtheorem{proposition}[theorem]{Proposition}

\newtheorem{example}[theorem]{Example}

\newenvironment{defn*}{\begin{definition}}{\end{definition}}

\newenvironment{proof}{\noindent{\bf Proof.}}{\hfill$\blacksquare$}

\newcommand{\vect}[1]{\mathbf{#1}}
 \def\figurename{Fig.}

\begin{document}
\title{Fusion of finite set distributions: Pointwise consistency and global cardinality}

\author{Murat \"{U}ney,~\IEEEmembership{Member,~IEEE}, J\'{e}r\'{e}mie Houssineau, Emmanuel Delande, Simon J. Julier, Daniel Clark
\thanks{This work was supported by the Engineering and Physical Sciences Research
Council (EPSRC) Grant number EP/K014277/1 and the MOD University Defence
Research Collaboration (UDRC) in Signal Processing.}
\thanks{Murat~\"{U}ney is with the Institute for Digital Communications, School of Engineering, University of Edinburgh, EH9 3FB, Edinburgh, UK (e-mail: m.uney@ed.ac.uk).}
\thanks{J\'{e}r\'{e}mie Houssineau is with the National University of Singapore, Department of Statistics and Applied Probability, National University of Singapore, Singapore 119077.}
\thanks{Emmanuel Delande is with the Institute for Computational Engineering and Sciences, University of Texas at Austin (edelande@ices.utexas.edu).}
\thanks{Simon J. Julier is with the Computer Science Department, University College London, London (e-mail: s.julier@ucl.ac.uk).}
\thanks{Daniel E. Clark is with D{\'e}partment CITI, Telecom-SudParis, 9, rue Charles Fourier
91011, EVRY Cedex, France (e-mail: daniel.clark@telecom-sudparis.eu).}
}


%


\maketitle

\begin{abstract}
A recent trend in distributed multi-sensor fusion is to use random finite set filters at the sensor nodes and fuse the filtered distributions algorithmically using their exponential mixture densities (EMDs). Fusion algorithms that extend covariance intersection and consensus based approaches are such examples. In this article, we analyse the variational principle underlying EMDs and show that the EMDs of finite set distributions do not necessarily lead to consistent fusion of cardinality distributions. Indeed, we demonstrate that these inconsistencies may occur with overwhelming probability in practice, through examples with Bernoulli, Poisson and independent identically distributed (IID) cluster processes. We prove that pointwise consistency of EMDs does not imply consistency in global cardinality and vice versa. Then, we redefine the variational problems underlying fusion and provide iterative solutions thereby establishing a framework that guarantees cardinality consistent fusion.
\end{abstract}


\markboth{IEEE Transactions on Aerospace and Electronic Systems}%
{\"{U}ney: Fusion of finite set distributions: Pointwise consistency and global cardinality}

%
\IEEEpeerreviewmaketitle

\section{Introduction}
\label{sec:Introduction}
\IEEEPARstart{I}{n} networked sensing, nodes perform local filtering and exchange filtered distributions as opposed to communicating raw measurements~\cite{Hall2013}. The problem of fusion is to find an estimate for the {\it a posteriori} distribution over some state space conditioned on two or more (conditionally) independent sensor data streams, given local posteriors computed by local filtering of each data stream individually.

A large body of work utilises exponential mixtures of distributions (EMDs) for fusion. These mixtures are found by taking the weighted geometric mean of their components followed by scaling to ensure integration to unity. They have been widely used for fusion of single object (probability) distributions~\cite{Julier2006a}. A well-known algorithm that utilises EMDs of Gaussian densities is covariance intersection~\cite{inproceedings:julier97c}. In covariance intersection (CI), the weights of the components in the mixture are selected using various criteria~\cite{Hurley2002}. The underlying variational problem considers minimising a cost that equals to the weighted sum of Kullback-Leibler divergences~\cite{Cover2006} of the fused density that is sought with respect to the mixture components. The stationary density and set of weights for this problem specifies an EMD which is deemed as a middle-ground of the components in a way analogous to logarithmic opinion pooling of experts~\cite{Heskes1998}.

The EMD form has been adopted for finite set densities in order to address fusion in the case of multiple objects~\cite{Mahler2000a}. Following the introduction of tractable recursive filters~\cite{Mahler2007} such as the probability hypothesis density (PHD) filter~\cite{Mahler2003}, and, explicit filtering algorithms using Gaussian mixture model (GMM) representations~\cite{Vo2006} and sequential Monte Carlo (SMC) techniques~\cite{Vo2005}, numerical algorithms that extend CI fusion to Bernoulli, PHD, and cardinalised PHD (C-PHD) were proposed~\cite{Uney2010, Uney2011,Uney2013}. These methods have been proved useful in improving localisation accuracy in multi-sensor problems including those involving heterogenous sensors~\cite{Barr2013}. 

Another utilisation of EMDs for fusion of finite set distributions has been within the network consensus framework~\cite{Olfati-Saber2007}. Briefly, iterative message passing algorithms which asymptotically compute the equally-weighted mixture, i.e., the (unweighted) geometric mean of the components, at all nodes of a sensor network are proposed for C-PHD~\cite{Battistelli2013}, multi-Bernoulli~\cite{Wang2017}, generalised MB~\cite{Jiang2016,Yi2017}, Bernoulli~\cite{Guldog2014}, and, labelled~\cite{Battistelli2015,Li2018} finite set filters. 

In~\cite{phdthesis:dabak92,Julier2006}, it has been proved that EMDs have a probability density that at no point in the state space overlooks the density of their components. This property is proposed as a working definition of consistency in the context of fusion~\cite{Julier2006}. Finite set density EMDs also satisfy this consistency condition pointwise, at every finite collection of points. 

Finite set distributions, on the other hand, factorise into a cardinality distribution on the number of objects and a localisation density conditioned on the cardinality~\cite{DaleyVere-Jones}. In this article, we show that the cardinality distributions of EMDs are not endowed with such consistency guarantees, in general. Such inconsistencies might result with smaller existence probabilities or estimates on the number of objects when the fused results are used instead of either of the inputs. This phenomena which might undermine the benefits of using diversity in sensing has been empirically observed by other researchers as well (see, e.g.,~\cite{Gunay2016}). Here, we provide explicit mathematical formulae specifying conditions under which the cardinality distributions of finite set EMDs are inconsistent. We demonstrate in examples that these inconsistencies are encountered sometimes with overwhelming probability under typical operating conditions and might lead to large discrepancies in, for example, the estimated number of objects and/or object existence probabilities. 

Based on these results, we argue that the variational problem needs to be decoupled for the cardinality and the localisation distributions (i.e., scaled Janossy~\cite{DaleyVere-Jones} distributions). Doing so separates the fusion of cardinality distributions and localisation terms. This approach results with the same localisation densities as the direct adoption of the variational problem, and, avoids any inconsistencies in the cardinality distribution. We show that pointwise consistency does not imply consistency in cardinality and vice versa. Then, we derive iterative algorithms for cardinality consistent fusion of finite set distributions.

The outline of the article is as follows: In Section~\ref{sec:problemdefinition}, we discuss fusion rules that accommodate EMDs in the light of the associated variational problems and pointwise consistency of EMDs. We provide our results regarding the cardinality inconsistencies of finite set EMDs in Section~\ref{sec:FiniteSetEMDsCardinalityDistributions}, together with examples. Then, we redefine the variational problem underpinning fusion and derive solutions for cardinality consistent fusion in Section~\ref{sec:CardinalityConsistentFusion}. Conclusions and future directions are provided in Section~\ref{sec:Conclusion}.

\section{Fusion as a variational problem}
\label{sec:problemdefinition}
\subsection{EMDs as weighted KLD centroids}
Given two probability density functions (PDFs) $f_i$ and $f_j$ on a state space $\cal X$, let us consider finding another density $f$ in the space of PDFs $\cal P$ such that $f$ captures the information contained in both of the input distributions. An intuitive approach which is geometric in flavour would involve finding the centroid of the input distributions based on a distance/divergence metric. Kullback-Leibler divergence~(KLD) is such a divergence metric which is used  in information geometry in a way similar to the squared Euclidean distance~\cite{Csiszar2004}, and, has an established relevance to estimation when ${\cal X}$ is a finite alphabet (which is often referred to as hypothesis testing)~\cite{Cover2006}.

The KLD of two distributions with densities $f$ and $g$ is computed as
\begin{equation}
 D(f||g) = \int_{\cal X} f(X) \log \frac{f(X)}{g(X)} \mathrm{d}X,
 \label{eqn:KLD}
\end{equation}
where $D$ is always nonnegative and vanishes for $f=g$.

Let us denote the centroid of $f_i$ and $f_j$ with respect to a weighted sum of KLD by $f_\omega$. This distribution is a solution to the associated variational problem given by
\begin{eqnarray}
  \text{(P)} \,\,\,\min_{f \in {\cal P} } J_{\omega}[f] &&  \nonumber \\
   J_{\omega}[f] &\triangleq& (1-\omega) D(f||f_i) + \omega D(f||f_j) \label{eqn:centroid} 
\end{eqnarray}
where $\omega \in [0,1]$ is a design parameter selecting the weight of the divergence of each point $f_i$ and $f_j$ in the space of probability distributions $\cal P$ over $\cal X$, with respect to $f$.

The solution to problem $\text{(P)}$ with the cost \eqref{eqn:centroid} is unique and found as
\begin{eqnarray}
 f_\omega(X) &=& \frac{1}{Z_\omega} f_i^{(1-\omega)}(X) f_j^\omega(X) \label{eqn:emd}\\
 Z_\omega &=& \int_{\cal X} f_i^{(1-\omega)}(X') f_j^\omega(X')\,\,\mathrm{d} X',
 \label{eqn:emdscale}
\end{eqnarray}
which can easily be seen after rearranging the cost in~\eqref{eqn:centroid} as
\begin{equation}
 J_\omega [f] = D(f || f_\omega ) - \log \int_{\cal X} f_i^{(1-\omega)}(X) f_j^\omega(X)\,\,\mathrm{d} X,
 \label{eqn:cost2}
\end{equation}
(see, for example, ~\cite[Eq.(3)]{Csiszar2003}), and, realising that the second term on the right hand side does not depend on $f$ (see Appendix~\ref{sec:statpointsproblemp} for a direct proof). In fact, this term is the scaled R\'enyi divergence~\cite{inproceedings:renyi60} of order $\omega$ from $f_j$ to $f_i$, i.e.,
\begin{eqnarray}
 J_\omega [f] &= &D(f || f_\omega ) - (\omega - 1) R_\omega ( f_j, f_i ),
  \label{eqn:cost3} \\
  R_\omega ( f_j, f_i ) &\triangleq& \frac{1}{\omega-1} \log \int_{\cal X} f_j^\omega(X) f_i^{(1-\omega)}(X) \mathrm{d} X, \notag \\
   &=& \frac{1}{\omega-1} \log Z_\omega. \notag
\end{eqnarray}

Let us consider the weight parameter $\omega$ as a free variable, and find the stationary point of $J_\omega$ in~\eqref{eqn:centroid} with respect to $\omega$ for $f=f_\omega$. For the case, the KLD term in~\eqref{eqn:cost3} vanishes and~\eqref{eqn:centroid} reduces to a cost function for finding the Chernoff information of $f_i$ and $f_j$~\cite{Cover2006} which is concave in $\omega$ \footnote{To be specific, in~\cite{chernoff1952}, Chernoff introduces $C(f_j, f_i) \triangleq -\log \min Z_\omega$ as a ``measure of divergence'' between two distributions. This quantity can equivalently be found by $ \max -\log Z_\omega$ in which the argument of maximisation is nothing but $J_\omega[f=f_\omega]$.}. In~\cite{Julier2006b}, it is explained that there is a unique stationary point $\omega^*$ which satisfies 
\begin{equation}
  D(f_\omega || f_i) = D( f_\omega || f_j)  \bigg\rvert_{\omega=\omega^*}.
 \label{eqn:StationaryOmega}
\end{equation}

The density $f_\omega$ in \eqref{eqn:emdscale} is obtained by normalising the weighted geometric mean of $f_i$ and $f_j$, and, thus  referred to as their geometric mean density (GMD), or, exponential mixture density (EMD). In this article we adopt the latter.

\subsection{Covariance intersection and generalisations}
\label{sec:genci}
The above discussion outlines a fusion algorithm which outputs the pair $(\omega^*,f_{\omega^*})$ using \eqref{eqn:StationaryOmega}~and~\eqref{eqn:emd} for fusing $f_i$ and~$f_j$. This can be rephrased as a $\max \min$ mathematical programme:
\begin{equation}
(\text{P2}) \,\,\,\,(\omega^*,f_{\omega^*}) \triangleq \arg \max_{\omega \in [0,1]} \min_{f \in {\cal P} } J_{\omega}[f].
  \label{eqn:centroid2}
\end{equation}

The input densities here are {\it a posteriori} in nature as they are propagated by local filters, i.e., they are conditioned on the data-streams of sensors $i$ and $j$, respectively. When $\cal X$ is $\mathbb{R}^d$, i.e., the $d$-dimensional space of real vectors, and, the distributions involved are Gaussians, this approach reduces to a set of linear algebraic operations which are known as the ``covariance intersection'' algorithm~\cite{inproceedings:julier97c}. In this setting, how well an approximation the EMD~\eqref{eqn:emd} is to the joint posterior\footnote{Here, we refer to the posterior distribution conditioned on the data streams of both sensors which is infeasible to compute given the limited communication and computational resources of the networked setting.} is studied in terms of bounds over the uncertainty spread characterised by covariance matrices (see, e.g.,~\cite{Chen2002,Reinhardt2015}). An information geometric characterisation of $f_{\omega^*}$ for multivariate Gaussians and other exponential family distributions is provided in~\cite{Nielsen2013} where it is proved that $f_{\omega^*}$ is the unique intersection point of the exponential geodesic curve joining $f_i$~and~$f_j$ (obtained by varying $\omega$ from $0$ to $1$ in~\eqref{eqn:emd}) and its dual hyperplane on the induced statistical manifold.

For general distributions, \eqref{eqn:emd},~\eqref{eqn:emdscale}~and~\eqref{eqn:StationaryOmega} are still valid as a solution to the variational fusion problem in~\eqref{eqn:centroid2}. The optimal weight selected through~\eqref{eqn:StationaryOmega} equates the cost in~\eqref{eqn:centroid} to the Chernoff information~\cite{Cover2006} between $f_i$ and $f_j$~\cite{Julier2006b}. Perhaps for this reason, some authors refer to this fusion rule as Chernoff fusion (see, for example~\cite{Chang2010} and the references therein).

\subsection{Other fusion rules utilising EMDs}
Other fusion methods that use EMDs include consensus based approaches as overviewed in Section~\ref{sec:Introduction}. These methods compute $f_\omega$ by iterative message passings between nodes. However, instead of finding stationary weights of the variational problem in~\eqref{eqn:centroid}, this network averaging approach can compute only an equally weighted EMD, and, when the number of iterations tends to infinity. Some other methods differ from the generalised CI approach described above in their weight selection criteria: Some authors argue that it might be more beneficial to select the value of $\omega$ in \eqref{eqn:centroid} that would maximise the ``peakiness'' of $f_\omega$~\cite{Mahler2000a}, or, to minimise the uncertainty captured by $f_\omega$ quantified by its Shannon differential entropy~\cite{Hurley2002}. 

\subsection{A notion of consistency in Fusion}
The uncertainty spread in EMDs of arbitrary distributions is characterised in terms of pointwise bounds. In~\cite{Julier2006}, the authors show that the scale factor in \eqref{eqn:emdscale} is less than or equal to one,~i.e., $ Z_\omega \leq 1$, and, consequently EMDs \eqref{eqn:emd} satisfy the following consistency condition:
\begin{equation}
 f_\omega(X) \geq \min \{ f_i(X),f_j(X) \}
 \label{eqn:consistency}
\end{equation}
for all points $X \in {\cal X}$ and $\omega \in [0,1]$. In other words, the fused distribution does not overlook the probability mass assigned by $f_i$ and $f_j$ onto the vicinity of any point in the state space. In this sense, this condition corresponds to a notion of consistency~\cite{Julier2006}, in the context of distributed fusion\footnote{Note that the use of the term ``consistency'' here differs from its use in classical statistics.}.

In this article, our concern is the consistency properties of EMDs of finite set distributions. These distributions  have been commonly used to represent multi-object scenes~\cite{Mahler2007}. The following discussion is valid for any fusion scheme that employs EMDs and random finite set (RFS) distributions in order to quantify uncertainty in, for example, ``the number of objects,'' (e.g., Poisson, i.i.d. cluster RFSs~\cite{Mahler2007}), ``existence probabilities'' (e.g., Bernoulli, multi-Bernoulli, generalised labelled MB RFSs~\cite{Vo2013} and MB mixtures~\cite{Williams2015}) irrespective of their weight selection mechanism. In the next section, we utilise~\eqref{eqn:consistency} for analysing finite set EMDs and examine the fused global cardinality distributions for inconsistencies and their consequences in estimating object existence probabilities and/or the number of objects.

\section{Finite set EMDs and cardinality distributions}
\label{sec:FiniteSetEMDsCardinalityDistributions}
In the case of finite set valued random variables, $\cal X$ is the space of finite subsets of $\mathbb{R}^d$ and the density $f$ is a set function characterised by i) a cardinality distribution with probability mass function (pmf) $p(n)$ over natural numbers $n=0,1,\ldots$, and, ii) localisation densities $\rho_n(x_1,...,x_n) $ for ${n=1,2,\ldots}$ which are symmetric in their arguments~\cite{DaleyVere-Jones}. The corresponding density has a set valued argument $X=\{x_1,\ldots,x_n \}$ and is given by
\begin{eqnarray}
 f(X) &=& p\left( n \right) \sum_{\sigma \in \Sigma_n } \rho_{n} (x_{\sigma(1)},...,x_{\sigma(n)}) \notag \\
 &=& p\left(  n \right) n! \rho_{n} (x_{\sigma'(1)},...,x_{\sigma'(n)})
 \label{eqn:RFSdensity}
\end{eqnarray}
where $n=|X|$ and $|.|$ denotes set cardinality. Here, $\Sigma_n$ is the set of all permutations of $(1,\ldots,n)$, and, $\sigma' \in \Sigma_n$ in the last line is an arbitrary permutation which is selected as the identity permutation in the rest of this article.

Note that $p$ in~\eqref{eqn:RFSdensity} sums to one and $\rho_n$s integrate to unity. The finite set density $f$ also integrates to one over $\cal X$,~i.e.,
\begin{equation}
 \int_{\cal X} f(X) \mu(\mathrm{d}X) = 1 \nonumber
\end{equation}
where $\mu$ is an appropriate measure. Let us select $\mu$ as
\begin{equation}
 \mu(\mathrm{d}X) = \sum_{n=0}^\infty \frac{\lambda_n(\mathrm{d}X \cap  {\cal X}_n )}{n!}
 \notag
\end{equation}
where ${\cal X}_n$ is the space of $n$-tuple of points in $\mathbb{R}^d$, and, $\lambda_n$ is the Lebesgue (volume) measure on ${\cal X}_n$~\footnote{Further details on the topic can be found in Section~II.B and Appendix~B in~\cite{Vo2005}, and, the references therein.}. An alternative form of this integral is referred to as the set integral~\cite{Mahler2007}, i.e.,
\begin{equation}
 \int_{\cal X} f(X) \mu(\mathrm{d}X) = \int_{ \mathbb{R}^d} f(X) \delta X, \nonumber
\end{equation}
where the right hand side is the set integral of $f$ defined~as~\footnote{Note that the set integral in~\eqref{eqn:setint} is defined for an arbitrary (measurable) function $f$, but, when $f$ is a finite point process density, \eqref{eqn:setint} is nothing but the total probability theorem applied on~\eqref{eqn:RFSdensity}~\cite{Baccelli2016}.}
\begin{eqnarray}
 \int_{ \mathbb{R}^d} f(X) \delta X 
 &\triangleq& \sum_{n=0}^\infty \frac{1}{n!} \int\limits_{{ \mathbb{R}^d} }\ldots \int\limits_{{\mathbb{R}^d} } f(\{x_{1},...,x_{n}\})\mathrm{d}x_1\dots\mathrm{d}x_n   \nonumber \\[-14pt] \label{eqn:setint} \\ 
 &= &\sum_{n=0}^\infty \int\limits_{{ \mathbb{R}^d} }\ldots \int\limits_{{\mathbb{R}^d} } p(n)\rho_{n} (x_{1},...,x_{n}) \mathrm{d}x_1\dots\mathrm{d}x_n. 
 \nonumber
\end{eqnarray}

Let us consider the EMD of finite set distributions $f_i$ and~$f_j$. For the case~\eqref{eqn:emd} is valid with the scale factor in~\eqref{eqn:emdscale} found using the set integral in~\eqref{eqn:setint}, i.e., 
\begin{equation}
 Z_\omega = \int_{ \mathbb{R}^d} f_i^{(1-\omega)}(X')f_j^\omega(X') \delta X'.
 \label{eqn:rfsemfscale}
\end{equation}
This scale factor is also less than one and consequently the finite set EMD satisfies the pointwise consistency condition in~\eqref{eqn:consistency} for every finite subset $X \subset \mathbb{R}^d$. 

In order to investigate the cardinality distribution of the EMD, let us substitute $f_i$ and $f_j$ in the form given in~\eqref{eqn:RFSdensity} into~\eqref{eqn:rfsemfscale}~and~\eqref{eqn:emd}, and, obtain the finite set EMD as
\begin{equation}
 f_\omega(X) =  p_\omega( n ) n! \rho_{\omega,n}(x_1,\ldots,x_n ) \nonumber
\end{equation}
where the localisation density for cardinality $n$ is  

\begin{multline}
\rho_{\omega,n}(x_1,\ldots,x_n ) \triangleq \frac{1}{z_\omega(n)}\rho_{i,n}^{(1-\omega)}(x_1,\ldots,x_n ) \\ \times \rho_{j,n}^{\omega}(x_1,\ldots,x_n ),
\label{eqn:emdloc}
\end{multline}
\begin{eqnarray}
 z_\omega(n) &=& \int\limits_{\mathbb{R}^d} \cdots \int\limits_{\mathbb{R}^d} \rho_{i,n}^{(1-\omega)}(x'_1,\ldots,x'_n ) \nonumber \\
 && \mspace{10mu} \times \rho_{j,n}^{\omega}(x'_1,\ldots,x'_n ) \mathrm{d}x'_1,\ldots,\mathrm{d}x'_n,
 \label{eqn:zomega}
\end{eqnarray}
and, the cardinality pmf is
\begin{eqnarray}
 p_\omega(n)&=& \frac{1}{N_\omega}p_i^{(1-\omega)}(n)p_j^{\omega}(n)z_{\omega}(n) \label{eqn:emdcard} \\
 N_\omega &=& \sum_{n'=0} p_i^{(1-\omega)}(n')p_j^{\omega}(n')z_{\omega}(n').
 \label{eqn:Nomega}
\end{eqnarray}
where $z_\omega(0)=1$ by convention, and for $n \neq 0$, $ z_\omega(n) < 1$ unless $\rho_{i,n}$ and $\rho_{j,n}$ are identical. The latter is a direct application of  H\"older's inequality (see, e.g., Theorem 188~in~\cite{Hardy1934}). It can be shown similarly that $N_\omega < 1$.

Let us focus on the fused cardinality pmf in \eqref{eqn:emdcard}. This distribution is {\it not an EMD} of the cardinality distributions of the components unlike the fused localisation distributions in~\eqref{eqn:emdloc} that are EMDs of the input localisation densities. In fact, the fused cardinality pmf is the scaled product of the cardinality EMD with the localisation density scale factors $z_\omega(n)$ in~\eqref{eqn:zomega}. As a result, the consistency property of EMDs does not apply to the fused cardinality distribution. Below, we first relate the consistency of the fused cardinality pmf to the sequence of scale factors and give a condition under which the fused cardinality distribution is inconsistent. Then, in the rest of this section, we demonstrate that inconsistent cardinality distributions occur under some typical operating conditions.
\begin{proposition}[Inconsistency in cardinality distribution]
\label{prop:inconsistency}
 Consider the fused cardinality pmf $p_\omega$ in~\eqref{eqn:emdcard},~\eqref{eqn:Nomega}. Consider the following {\it inconsistency} condition for $p_\omega(n)$  obtained by negating the consistency condition:
 \begin{equation}
  p_\omega(n)< \min\{\,\, p_i(n),\, p_j(n) \,\, \}.
  \label{eqn:inconsistency}
 \end{equation}
 This condition holds true if
 \begin{equation}
  z_\omega(n)<\frac{ \sum_{n' \neq n} p_i^{(1-\omega)}(n')p_j^{\omega}(n') z_\omega(n') }{ \frac{ p_i^{(1-\omega)}(n)p_j^{\omega}(n)}{ \min\{\,\, p_i(n),\, p_j(n) \,\, \} }  -  p_i^{(1-\omega)}(n)p_j^{\omega}(n) },
  \label{eqn:condition}
 \end{equation}
where $z_\omega(n)$ is given in \eqref{eqn:zomega}.
\end{proposition}

The proof is given in Appendix~\ref{sec:proof}. The above proposition points out that the fused cardinality distribution opts to disagree with local results on the probability of number of objects when the $n$th localisation scale $ z_\omega(n)$ is comparably small. This is in stark contrast with the fused localisation densities in~\eqref{eqn:emdloc} which always satisfy the consistency condition
\begin{equation}
 \rho_{\omega,n}(x_1,\ldots,x_n ) \geq \min \{\rho_{i,n}(x_1,\ldots,x_n ),\rho_{j,n}(x_1,\ldots,x_n )  \}
 \nonumber
\end{equation}
for all $x_1,\ldots,x_n \in \mathbb{R}^d$ and $n$, as they are EMDs.

The scale factors modulating the cardinality pmf, i.e., $z_\omega(n)$, are found by taking the inner products of the input localisation densities raised to fractional powers. As explained above, these terms are upper bounded by one with $z_\omega(n)$ equaling unity only when $\rho_{i,n}$ and $\rho_{j,n}$ are equal (see~\cite{Julier2006} for an alternative proof). In fusion networks, however, one of the main goals is to benefit from sensing diversity which means $\rho_{i,n}$ and $\rho_{j,n}$ will have a comparably small overlap in their confidence regions. As a result, much smaller $z_\omega(n)$ values should be expected in typical operating conditions.

Now, let us consider some particular RFS families and demonstrate the consequences of Proposition~\ref{prop:inconsistency}.

\subsection{Bernoulli finite set EMDs and fused existence probabilities}
\label{sec:Bernoulli}
Bernoulli finite set distributions select at most one object from a population. Collections, and mixtures thereof are used to represent multi-object models the fusion of which reduces to EMD fusion of Bernoulli pairs~(see, e.g., \cite{Jiang2016,Yi2017}). For a Bernoulli finite set, the cardinality pmf in~\eqref{eqn:RFSdensity} is given by
\begin{equation}
 p(n) = \begin{cases}
    1-\alpha,& \,  n= 0,\\
    \alpha,& \,    n=1,\\
    0,& \,    \text{otherwise}
    \end{cases}
    \label{eqn:BernoulliCardinality}
\end{equation}
where the parameter $\alpha$ is referred to as {\it the existence probability} of the object modelled. 

There is also a single localisation density $\rho_n$ for $n=1$ which we will denote by $\rho$. Therefore, given two Bernoullis $f_i=( \alpha_i, \rho_i )$ and $f_j=( \alpha_j, \rho_j )$, the sequence $z_\omega(n)$ reduces~to
\begin{equation}
 z_\omega(n) = \begin{cases}
    1,& \,  n= 0,\\
     z_\omega \triangleq \int_{\mathbb{R}^d} \rho_i^{(1-\omega)}(x)\rho_j^\omega(x)\mathrm{d}x ,& \,    n=1,\\
    0,& \,    \text{otherwise}.
    \end{cases} 
    \label{eqn:zomegabernoulli}\\
\end{equation}

\begin{corollary}
 \label{ref:corbernoulli}
 The inconsistency condition given by Proposition~\ref{prop:inconsistency} for Bernoulli finite set distributions reduces to that the existence probability of the EMD given by~\cite{Uney2013}
\begin{equation}
 \alpha_\omega = \frac{ \alpha_i^{(1-\omega)}\alpha_j^{\omega}z_\omega }
 { (1-\alpha_i)^{(1-\omega)} (1-\alpha_j)^\omega + \alpha_i^{(1-\omega)} \alpha_j^\omega z_\omega }
 \label{eqn:BernoulliEMDexisprob}
\end{equation}
is smaller than either of $\alpha_i$ or $\alpha_j$ if
\begin{equation}
 z_\omega < \frac{  (1-\alpha_i)^{(1-\omega)} (1-\alpha_j)^\omega  }{ \alpha_i^{(1-\omega)} \alpha_j^\omega/\min\{\,\alpha_i\,,\,\alpha_j\,\} - \alpha_i^{(1-\omega)} \alpha_j^\omega  }.
 \nonumber
\end{equation}
\end{corollary}

The proof follows from substituting the sequence~\eqref{eqn:zomegabernoulli} in~Proposition~\ref{prop:inconsistency}, and, in particular in~\eqref{eqn:inconsistency}~and~\eqref{eqn:condition}. This condition is very often satisfied in sensing applications as explained before. For example, if $\alpha_i$ and $\alpha_j$ are equal, then this condition reduces to $z_\omega < 1$ which always holds  for all practical purposes as $\rho_i$ and $\rho_j$ should not be expected to be identical. For  $\alpha_i \neq \alpha_j$, this inconsistency still occurs with overwhelming probability in Bernoulli fusion which is demonstrated in the following example. 
\begin{figure}[b!]
\vspace{-5pt}
  \centering{
  \begin{minipage}{\linewidth}
    \includegraphics[width=0.19\linewidth]{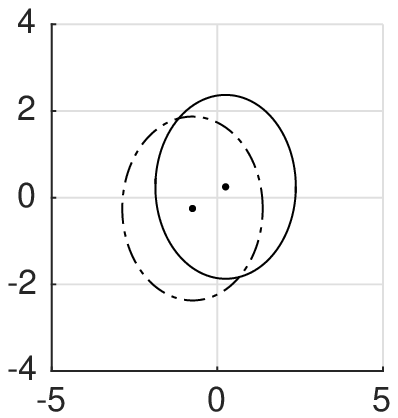}
    \includegraphics[width=0.19\linewidth]{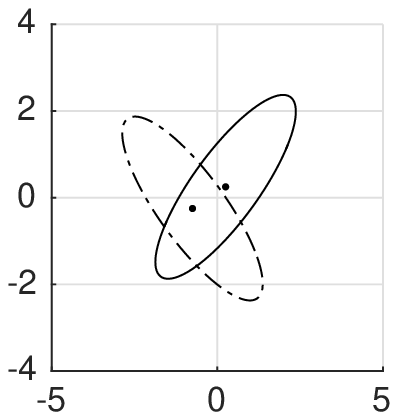}
    \includegraphics[width=0.19\linewidth]{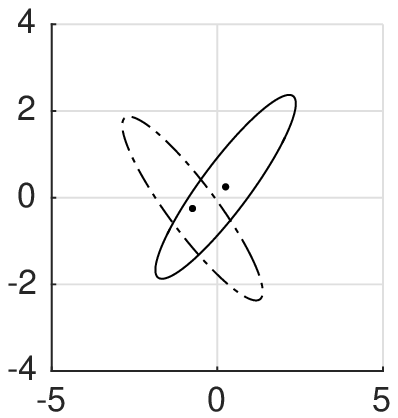}
    \includegraphics[width=0.19\linewidth]{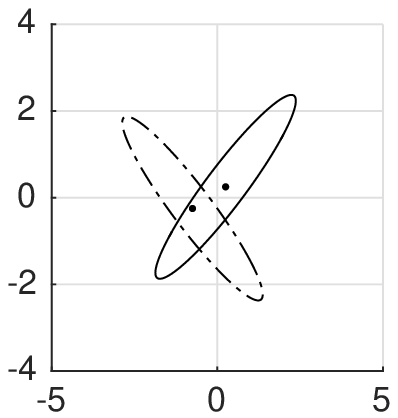}
    \hfill
    \includegraphics[width=0.19\linewidth]{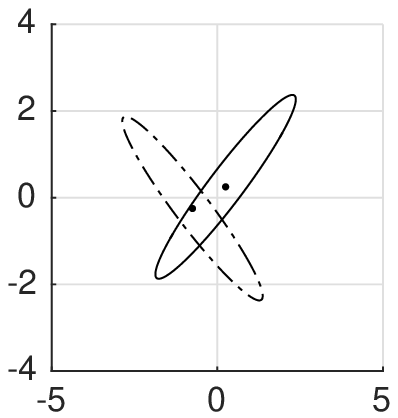}
  \end{minipage}
  \vspace{-4pt}
  }
  \caption[Gaussians]{Localisation densities of the Gauss-Bernoulli finite sets, i.e., $\rho_i$ (solid line)~and~$\rho_j$ (dash-dotted line), in Example~\ref{ex:GaussBernoulli} for increasing sensing diversity as the covariance condition number is increased as $\kappa=1,10,20,30,40$ (left to right).}
  \label{fig:gausses}
\end{figure}
\begin{example}[Gauss-Bernoulli EMDs]
\label{ex:GaussBernoulli}
\begin{figure}[t!]
\vspace{-5pt}
  \centering{
  \includegraphics[width=\linewidth,height=36mm]{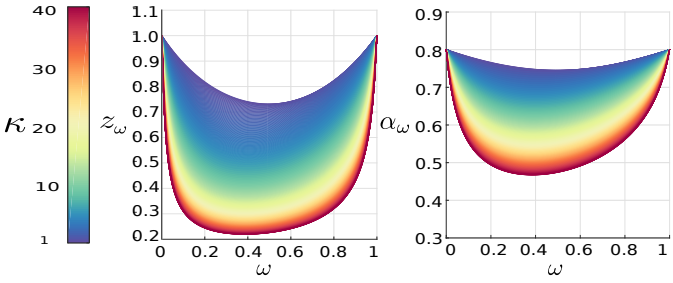}
  \vspace{-20pt}
  }
  \caption[Gauss-Bernoulli results]{The scale factor (left) and the fused existence probabilities in the Gauss-Bernoulli Example~\ref{ex:GaussBernoulli}. These quantities are calculated for varying sensing diversity $\kappa$ (equivalently, the covariance condition number of the Gaussians) and (exponential) mixture weight $\omega$ values. }
  \label{fig:gbcombo}
\end{figure}
Let us consider Bernoulli distributions with Gaussian localisation densities given~by
\begin{equation}
 \rho_i(x) = {\cal N}(x; \vect{m}_i,\vect{C}_i ),\,\,\,\,\rho_j(x) = {\cal N}(x; \vect{m}_j,\vect{C}_j ),
\label{eqn:GaussianLocs}
\end{equation}
where $\vect{m}$ is the mean vector and $\vect{C}$ is the covariance matrix. The fused localisation density $\rho_\omega$ for the case is a Gaussian with mean and covariance given by
\begin{eqnarray}
\vect{m}_\omega &= &\vect{C}_\omega\big((1-\omega)\vect{C}_i^{-1}\vect{m}_i + \omega \vect{C}_j^{-1}\vect{m}_j \big)
\label{eqn:momega} \\
\vect{C}_\omega &=& \left( (1-\omega)\vect{C}_i^{-1} + \omega \vect{C}_j^{-1} \right)^{-1}.
\label{eqn:Comega}
\end{eqnarray}
 The scale factor $z_\omega$ is found using integration rules for Gaussians as
\begin{multline}
\mspace{-20mu} z_\omega = \frac{ \left| \vect{\vect{C}}^{-1}_i \right|^{ (1-\omega)/2 } \left| \vect{C}^{-1}_j\right|^{\omega/2} } { \left| \vect{C}_\omega \right|^{1/2}  } \exp \left\{ -\frac{1}{2}\bigg( (1-\omega)\vect{m}_i^T \vect{\vect{C}}_i^{-1}\vect{m}_i \bigg. \right. \\
\left. \bigg. + \omega \vect{m}_j^T\vect{\vect{C}}_j^{-1}\vect{m}_j
- \vect{m}_\omega^T\vect{C}_\omega^{-1}\vect{m}_\omega \bigg)\right\}.
\label{eqn:GaussEMDscale}
\end{multline}

Let us consider two Bernoullis with existence probabilities $\alpha_i = \alpha_j = 0.8$ with localisation densities of mean vectors $\vect{m}_i = [0.25,0.25 ]^T$ and $\vect{m}_j=[-0.75,-0.25 ]^T$, respectively, where $(.)^T$ denotes vector transpose. We select the covariance matrices as rotated versions of a diagonal covariance given by  
\begin{eqnarray}
 \vect{C}_i &=  &\vect{R}(\pi/4) \vect{\Sigma}\vect{R}^T(\pi/4), \notag \\
 \vect{C}_j &=  &\vect{R}(-\pi/4) \vect{\Sigma}\vect{R}^T(-\pi/4), \notag \\
   \vect{\Sigma}&=&\left[ \begin{array}{cc} \sigma_1^2,& 0 \\
   0,&\sigma_2^2 \end{array} \right], \notag \\
   \vect{R}(\phi) &=& \left[ \begin{array}{rr} \cos\phi,& -\sin \phi \\
   \sin \phi,& \cos \phi \end{array} \right].
   \notag
\end{eqnarray}
This covariance structure is typical with sensors placed at different positions and taking their measurements from different aspect angles of the surveillance zone. The condition number of $\Sigma$ -- equivalently, $\vect{C}_i$ and $\vect{C}_j$ -- is given by $\kappa=\sigma_1^2/\sigma_2^2$ and has higher values for sensors with range/cross-range ambiguity such as cameras/radars. We vary this quantity from $\kappa=1$ to $40$. \figurename~\ref{fig:gausses} depicts the uncertainty ellipses of sample Gaussians by using three times the standard deviation along the eigen vector directions.

The behaviours of the fused existence probability in \eqref{eqn:BernoulliEMDexisprob} and the scale factor in~\eqref{eqn:GaussEMDscale} are our concern. \figurename~\ref{fig:gbcombo} presents both the $z_\omega$ and $\alpha_\omega$ values obtained by varying the condition number $\kappa$ with small steps from $1$ to $40$ hence increasing the sensing diversity. The exponential mixture weights $\omega$ take values from a dense grid over $[0,1]$. As pointed out in this section, the scale factor values are always smaller than unity, and, can often take very small values. The scale factor monotonically decreases with $\kappa$ which controls the sensing diversity. It is convex with respect to the mixture weight $\omega$, as pointed out in Section~\ref{sec:problemdefinition}.

The fused existence probabilities given in \figurename~\ref{fig:gbcombo} demonstrate the inconsistency in cardinality. In this example, this quantity is always smaller than the input existence probabilities admitting inconsistency for all selections of $\kappa$~and~$\omega$. Moreover, the fused existence probability drops below $0.5$ for large values of the sensing diversity parameter~$\kappa$. This threshold is often used as the Bayesian decision boundary for detection and despite that the input sources are fairly confident on the existence of an object with existence probabilities of $\alpha_i = \alpha_j = 0.8$, detection might be missed if based on the fused result instead, thereby undermining the benefits of sensing diversity. As a result, the inconsistency in cardinality may lead to inconsistency in decision making when EMDs of finite set distributions are used.{\hfill$\blacksquare$}
\end{example}

\subsection{EMDs of Poisson finite set distributions}
\label{sec:Poisson}
Poisson finite set densities are capable of representing many objects and underpin popular multi-object filters such as the PHD filter~\cite{Mahler2003}. Their cardinality pmf in~\eqref{eqn:RFSdensity} is given by a Poisson distribution, i.e.,
\begin{equation}
 p(n) = \frac{ \mathrm{e}^\lambda \lambda^n }{n!}
 \label{eqn:Poissonpmf}
\end{equation}
where $\lambda$ is the expected number of objects. The localisation densities factorise over the density for $n=1$~as
\begin{equation}
 \rho_{n}(x_1,\ldots,x_n) = \prod_{i=1}^{n} \rho_1(x_i),
 \label{eqn:Poissonloc}
\end{equation}
making it possible to parameterise the entire finite set distribution with a scalar and a single density\footnote{We drop the subscript in $\rho_1$ for the rest of this subsection and denote it by~$\rho$.}.

For two Poissons $f_i=( \lambda_i, \rho_i )$ and $f_j=( \lambda_j, \rho_j )$, the sequence $z_\omega(n)$ is a geometric sequence found by subsituting from~\eqref{eqn:Poissonloc} for both $i$~and~$j$ into~\eqref{eqn:zomega}. This sequence is found~as
\begin{eqnarray}
z_\omega(n) &=& z_\omega^n  \label{eqn:zomegaPoisson} \\
  z_\omega &\triangleq& z_\omega(1) = \int_{{\mathbb{R}^d}} \rho_i^{(1-\omega)}(x')\rho_j^\omega(x')\mathrm{d}x,
  \label{eqn:commonratio}
\end{eqnarray}
where $z_\omega<1$ unless $\rho_i$ and $\rho_j$ are identical, as aforementioned.

The expected number of objects with respect to an EMD with weight parameter $\omega$ is given by~\cite{Uney2013}
\begin{equation}
 \lambda_\omega = \lambda_i^{(1-\omega)}\lambda_j^\omega  z_\omega.
 \label{eqn:PoissonEMDLambda}
\end{equation}
\begin{proposition}[Poisson inconsistency in expectation]
\label{prop:PoissonInconsistency}
Let us consider an inconsistency condition for Poisson cardinality distributions in terms of their expectations:
 \begin{equation}
 \lambda_\omega < \min\{ \lambda_i, \lambda_j \}.
 \label{eqn:PoissonIcondition}
\end{equation}

 This condition holds whenever
\begin{equation}
 z_\omega < \frac{\min\{ \lambda_i, \lambda_j \}}{\max\{ \lambda_i, \lambda_j \}}.
 \label{eqn:PoissonIC}
\end{equation}
\end{proposition}

The proof follows easily from substituting~\eqref{eqn:PoissonIC} in~\eqref{eqn:PoissonEMDLambda}~and~\eqref{eqn:PoissonIcondition}. It is instructive to contrast this result with Proposition~\ref{prop:inconsistency}. The latter holds for any class of finite set densities and considers their cardinality distributions for different $n$. The above result is on {\it the expected value} of $n$ in Poisson finite set densities. The condition in~\eqref{eqn:PoissonIC} is satisfied with overwhelming probability in practice leading to inconsistencies as observed, for example, in~\cite{Gunay2016}. For example, for $\lambda_i=\lambda_j=\lambda$, this reduces to the common ratio $z_\omega$ being less than one which should --as previously discussed-- always be expected to be the case in practice. 

The inconsistency in decision making for the case is related to the estimation of the number of objects. In Poisson finite set models, the minimum mean squared error (MMSE) estimation principle is used which leads to the use of $\lambda$ as the estimated number of objects\footnote{{\it Maximum a posteriori} (MAP) estimation is not used with Poisson cardinality distributions as~\eqref{eqn:Poissonpmf} is not guaranteed to have a unique maximum. Notice that, for example,~\eqref{eqn:Poissonpmf} evaluates at the same value for both $n=0$ and $n=1$ for $\lambda=1$.}. As a result, the EMD density always underestimates the number of objects despite that the source densities might be consistently suggesting otherwise,  in practice. The magnitude of the error stemming from this bias depends on the value of $z_\omega$.

\subsection{EMDs of IID cluster finite set densities}
\label{sec:iid}
IID cluster finite set distributions relax the Poisson cardinality pmf in~\eqref{eqn:Poissonpmf} and take arbitrary cardinality pmfs underpinning the C-PHD filter~\cite{Mahler2007a}. The localisation densities still take the factorised form in~\eqref{eqn:Poissonloc} leading to the identical geometric series $z_\omega(n)$ in~\eqref{eqn:zomegaPoisson}. For the case, Proposition~\ref{prop:inconsistency} specialises as follows:
\begin{figure}[t!]
  \centering{
  \subfloat[]{ \includegraphics[width=0.49\linewidth]{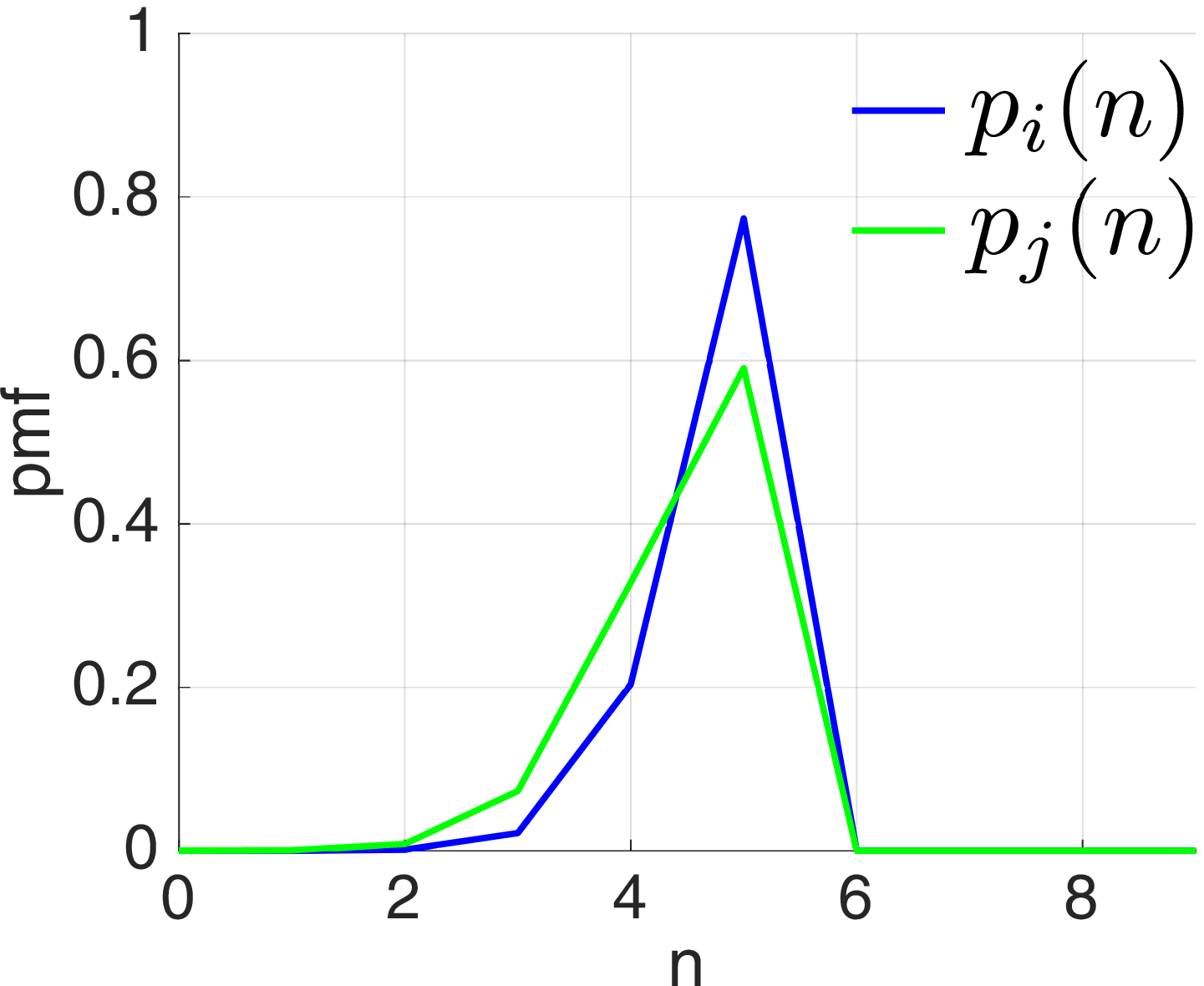} \label{fig:BB} }
  \subfloat[]{ \includegraphics[width=0.49\linewidth]{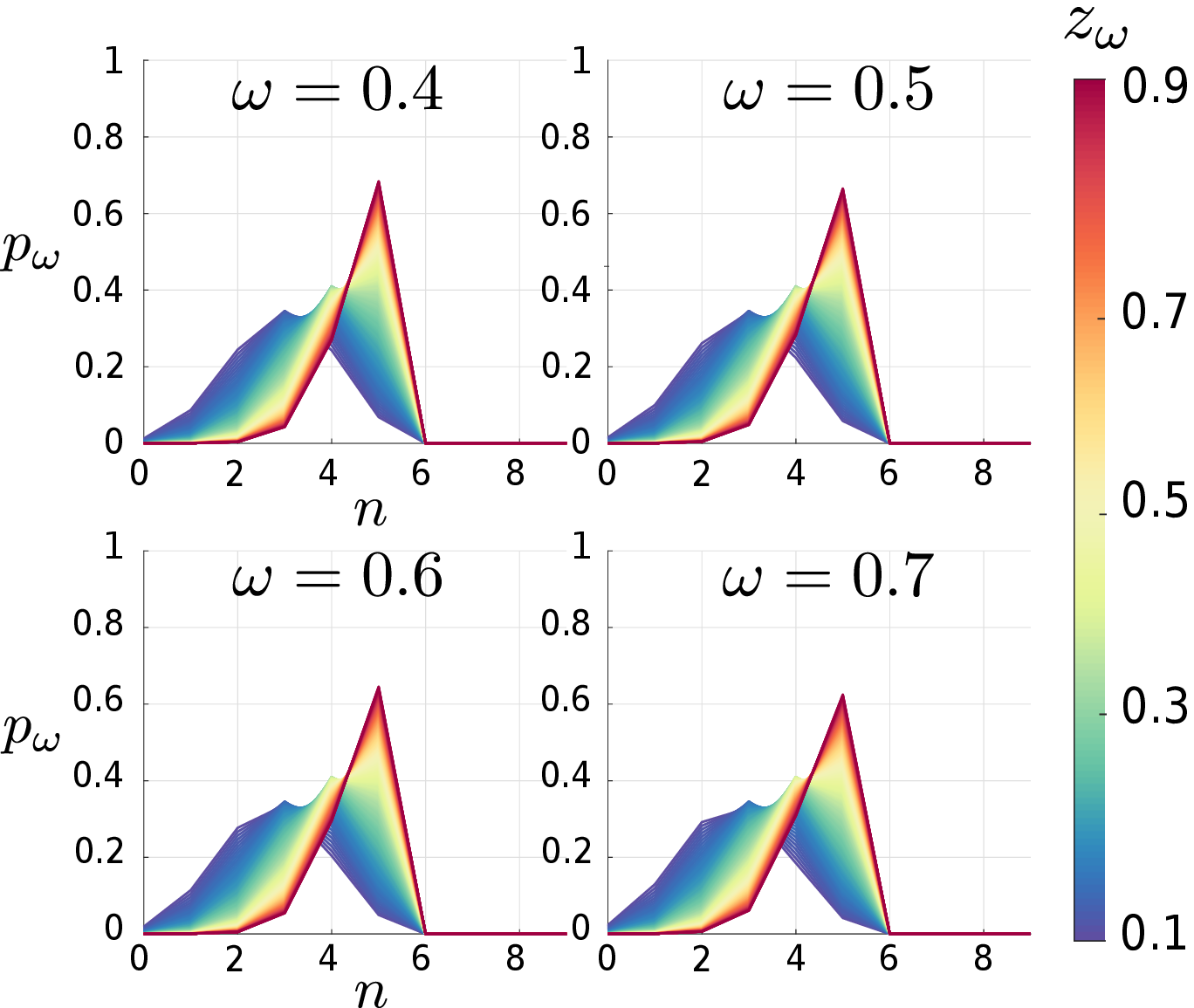} \label{fig:BBfusion} }\\
  \subfloat[]{ \includegraphics[width=0.49\linewidth,height=36mm]{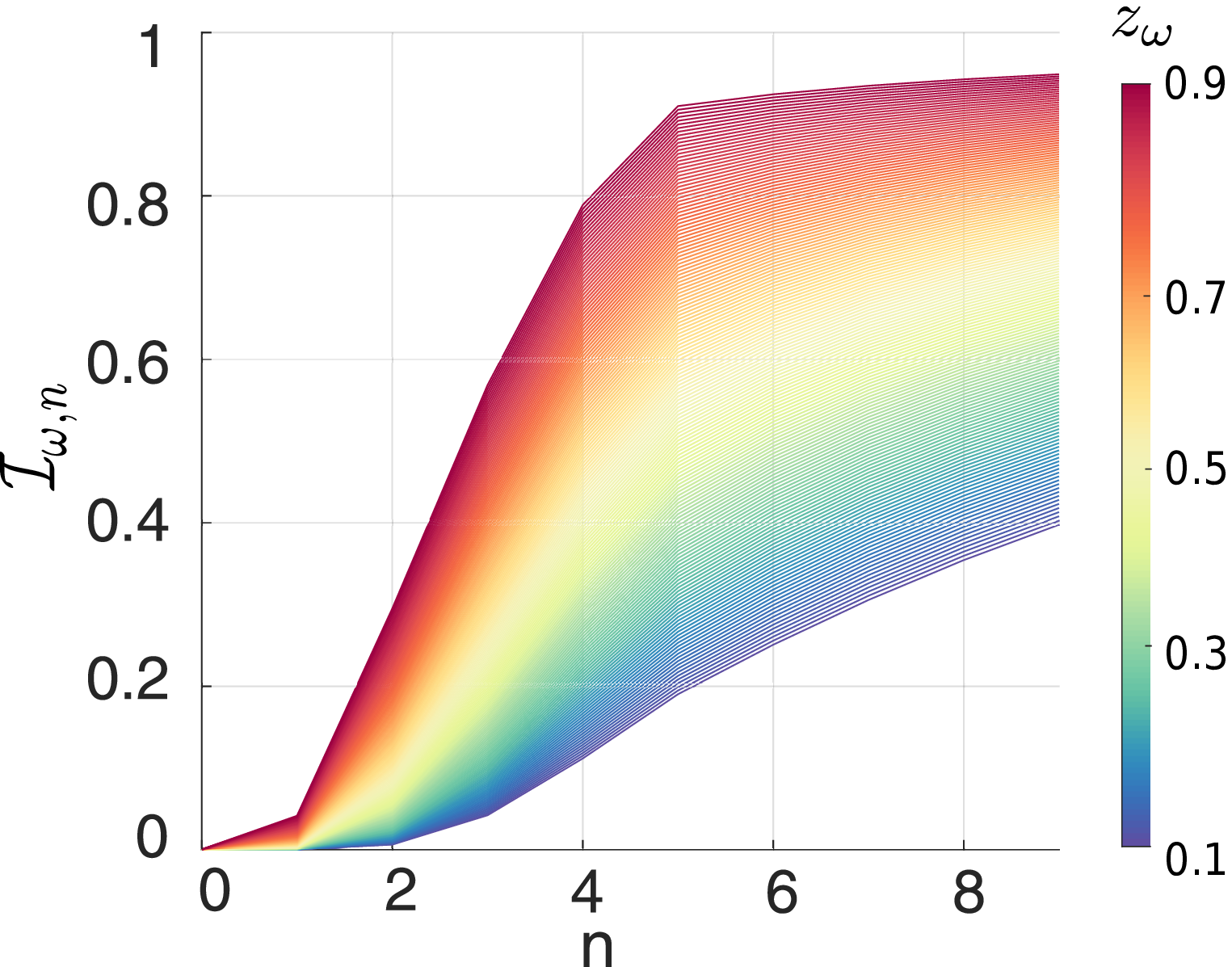} \label{fig:calI} }
  \subfloat[]{ \includegraphics[width=0.49\linewidth,height=36mm]{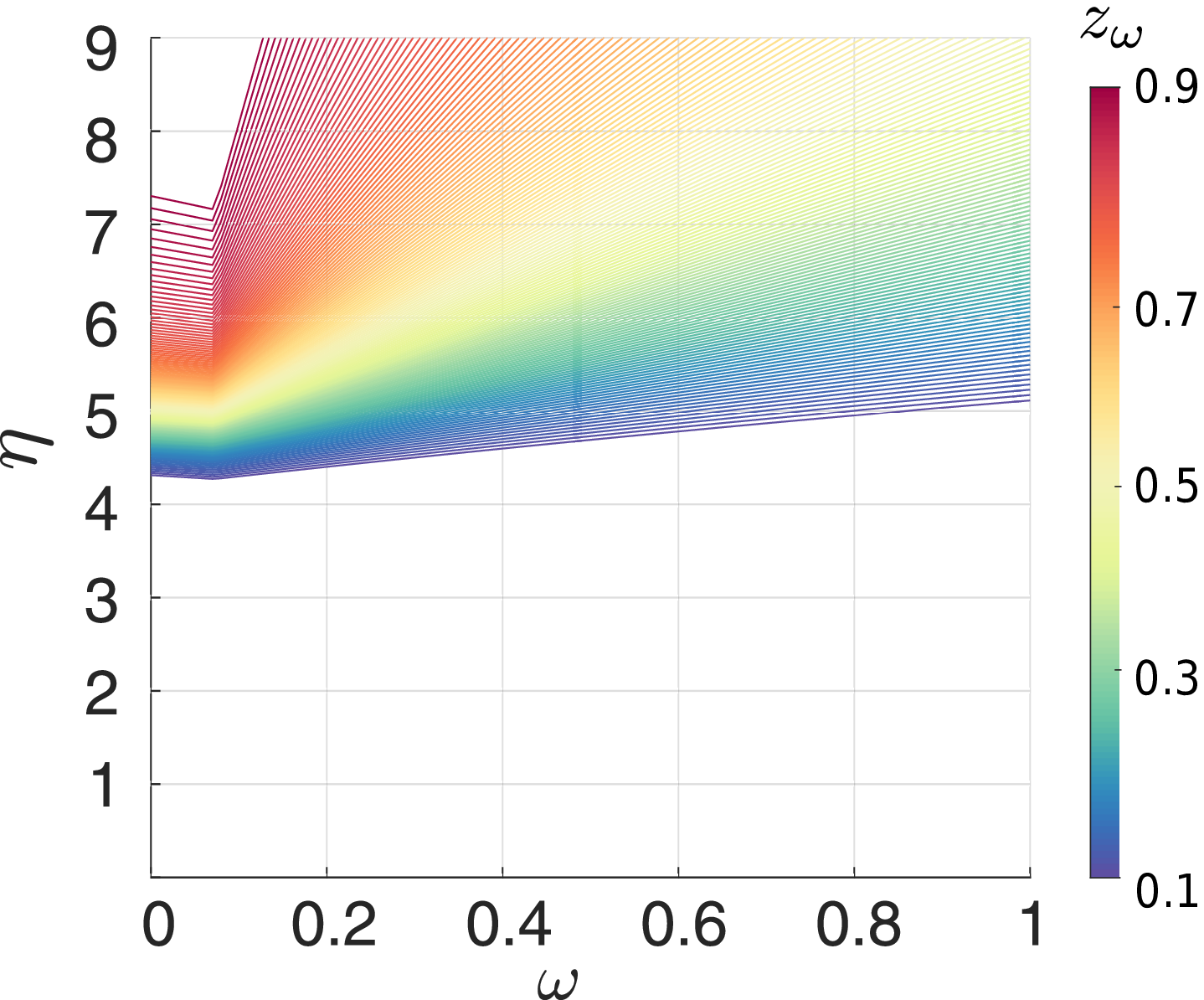} \label{fig:eta} }\\
  \subfloat[]{ \includegraphics[width=0.49\linewidth,height=36mm]{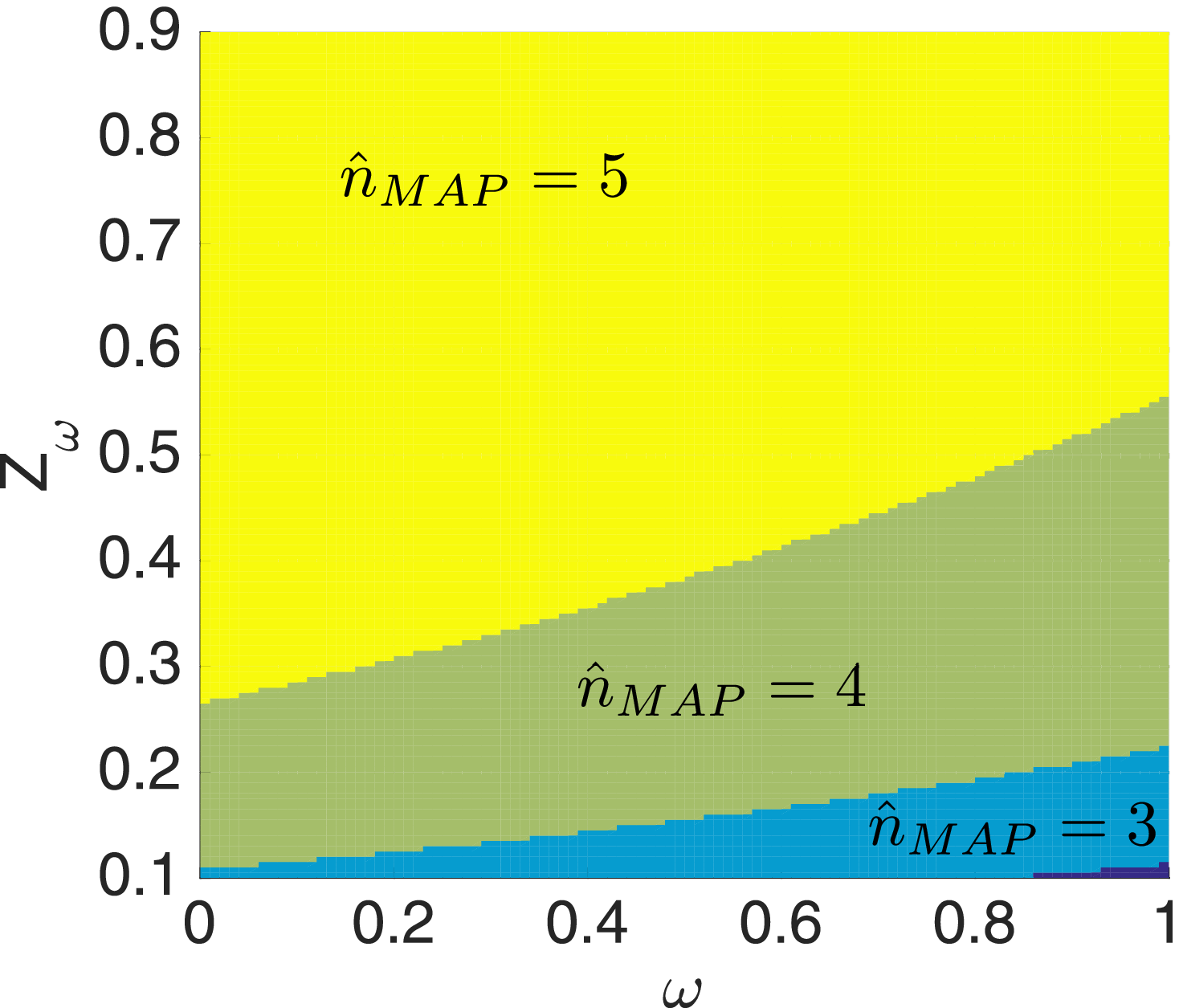} \label{fig:mapest1} }
  \subfloat[]{ \includegraphics[width=0.49\linewidth]{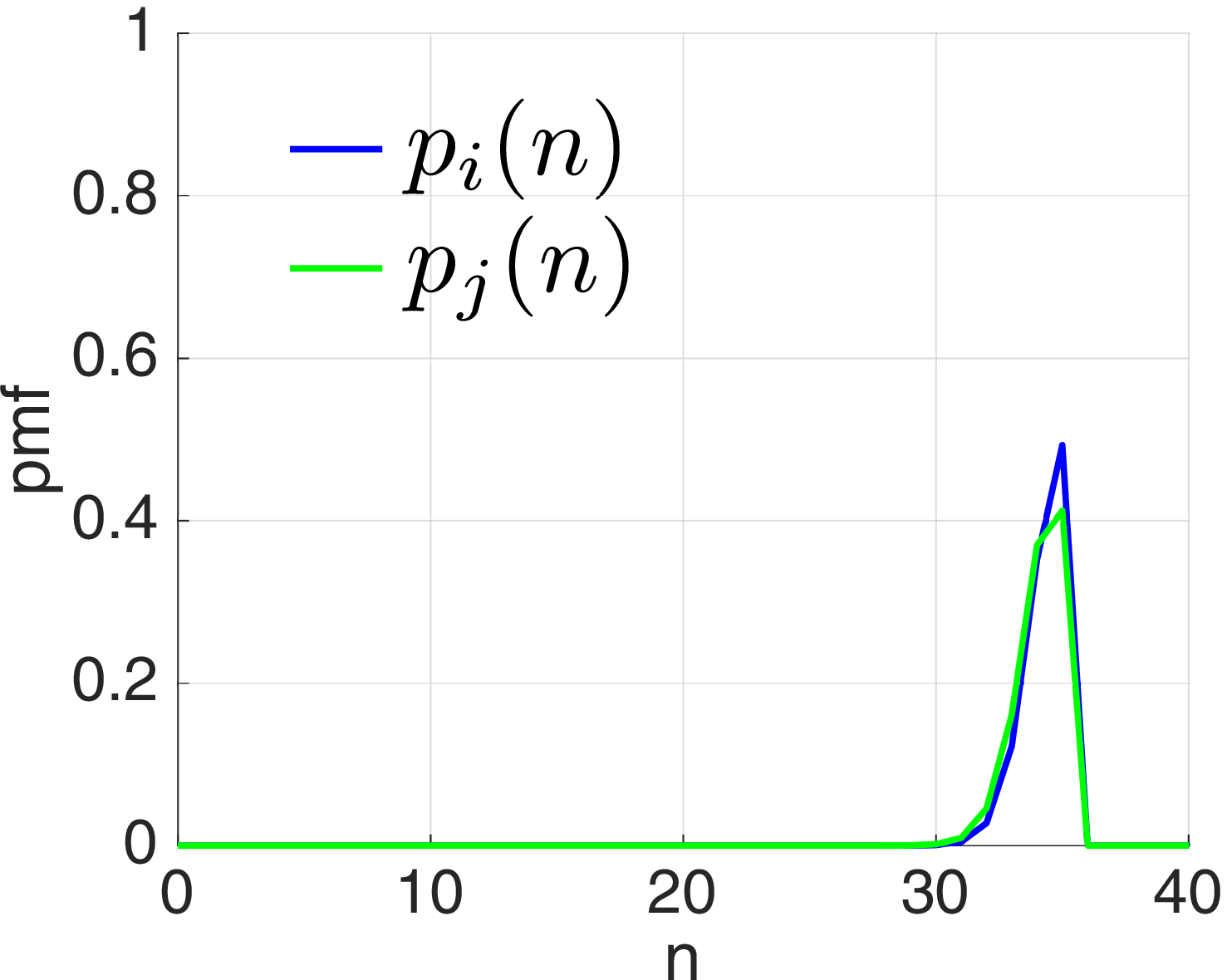} \label{fig:cardHigh} }\\
  \subfloat[]{ \includegraphics[width=0.49\linewidth]{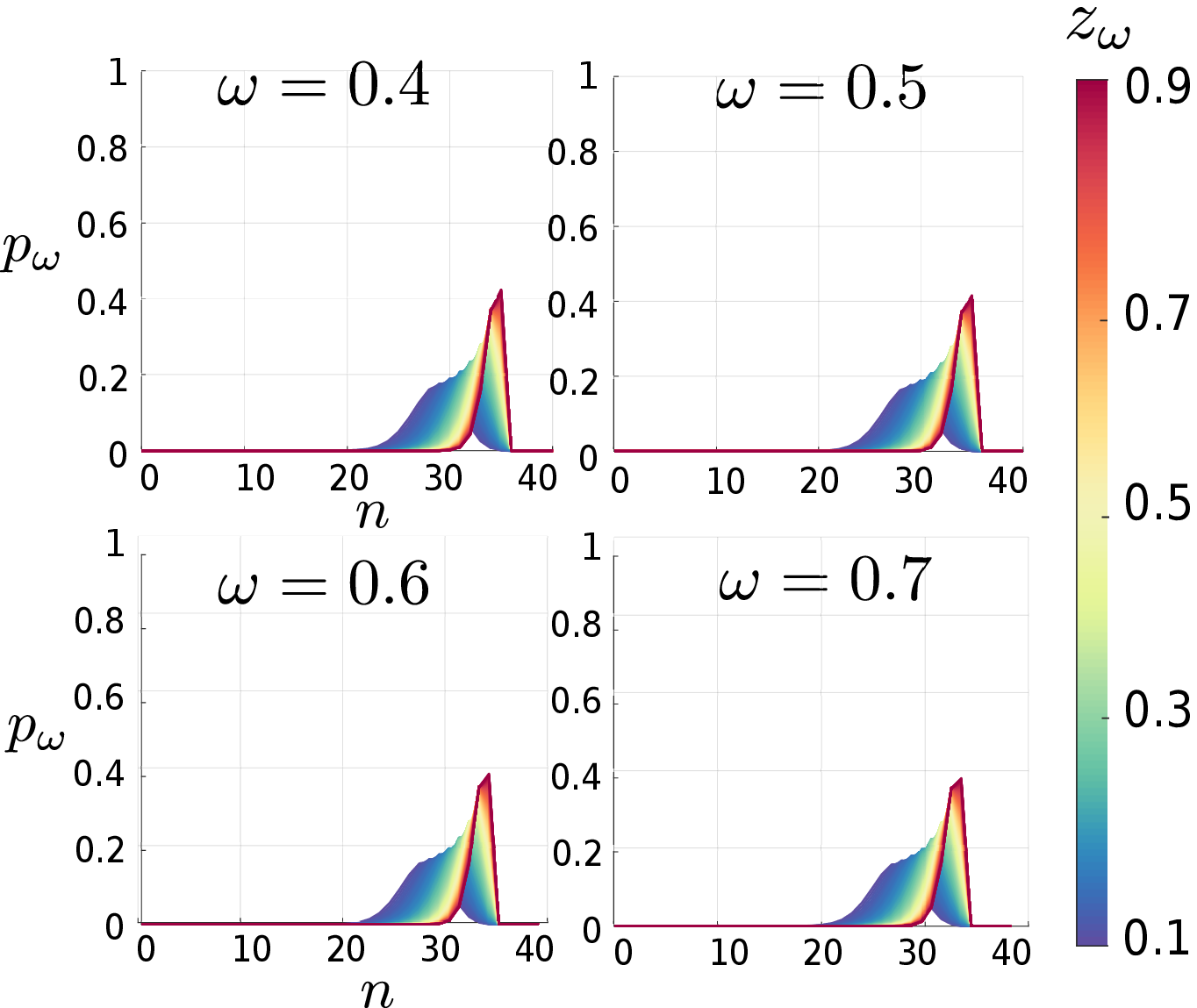} \label{fig:BBfusionHigh} }
  \subfloat[]{ \includegraphics[width=0.49\linewidth,height=36mm]{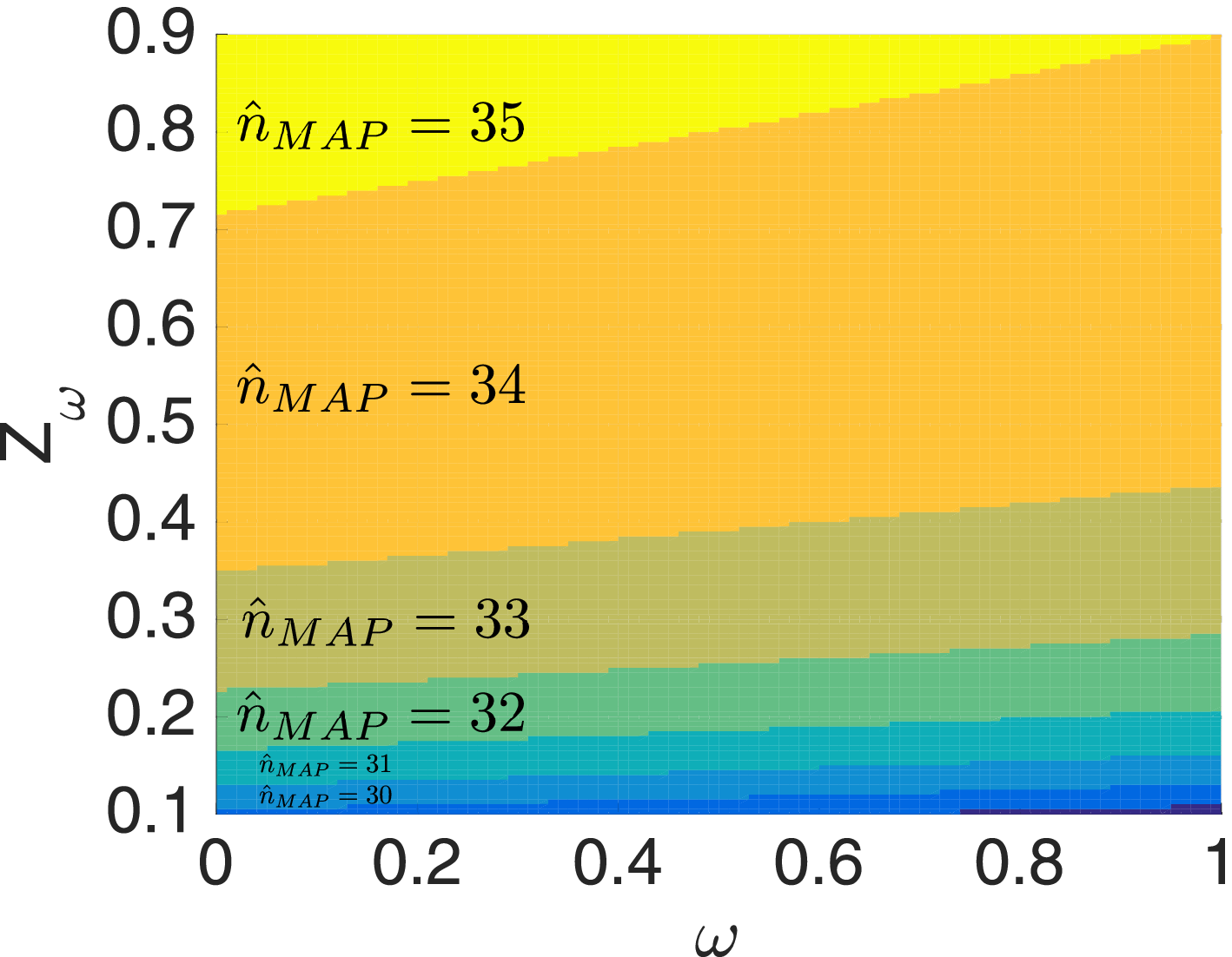} \label{fig:mapBBHigh} }
  }
  \caption[Bernoulli-Bernoulli fusion]{Illustrations of the results in Example~\ref{sec:IIDClusterExample}: \subref{fig:BB} Two cardinality distributions peaking at $n=5$. \subref{fig:BBfusion} Fused cardinalities for some intermediate values of $\omega$ and $0.1 \leq z_\omega \leq 0.9$. \subref{fig:calI} The inconsistency upper bound in \eqref{eqn:IIDinconsistency}. \subref{fig:eta} The inconsistency threshold in \eqref{eqn:consistencycutoff}. \subref{fig:mapest1} MAP estimates using the fused cardinalities for varying $\omega$ and $z_\omega$. \subref{fig:cardHigh} Cardinality distributions peaking at $n=35$. \subref{fig:BBfusionHigh}  Fused cardinalities for some intermediate values of $\omega$ and $0.1 \leq z_\omega \leq 0.9$. \subref{fig:mapBBHigh} MAP estimates using the fused cardinalities for varying $\omega$ and $z_\omega$.  }
  \label{fig:iidinconsistency}
\end{figure}
\begin{corollary}[IID cluster inconsistency]
\label{cor:iid}
 Given two IID cluster finite set distributions $f_i=( p_i(n), \rho_i )$ and $f_j=( p_j(n), \rho_j )$, the fused cardinality distribution $p_\omega$ satisfies the inconsistency condition in~\eqref{eqn:inconsistency} in Proposition~\ref{prop:inconsistency} for the number of objects $n$ and non-zero $p_i(n),p_j(n)$ if
 \begin{equation}
 z_\omega < {\cal I}_{\omega, n} \nonumber \\
\end{equation}
holds, where the term on the right hand side is
\begin{equation}
  {\cal I}_{\omega, n} \triangleq \left( N_\omega \frac{\min\{ p_i(n), p_j(n) \}}{p^{(1-\omega)}_i(n)p^{\omega}_j(n)} \right)^{1/n},  \label{eqn:IIDinconsistency} 
\end{equation}
$z_\omega$ is given in \eqref{eqn:commonratio}, and, $N_\omega$ is obtained by substituting \eqref{eqn:zomegaPoisson} in~\eqref{eqn:Nomega}.
\end{corollary}
The proof follows from substituting \eqref{eqn:zomegaPoisson} in~\eqref{eqn:condition}  and using \eqref{eqn:Nomega} after rearrangement of the terms. The inconsistency condition in \eqref{eqn:IIDinconsistency} depends both on $n$ and $\omega$, and, it is not straightforward to relate the inconsistent bins to object number estimation either in the MMSE or MAP rules. On the other hand, the base of the exponent in \eqref{eqn:IIDinconsistency} is smaller than one and hence $\cal I$ approaches to one as $n$ grows. Therefore, for some threshold $\eta$, the fused object number probabilities will be lower than the input cardinality reports for all $n>\eta$. Such a threshold can easily be found from~\eqref{eqn:IIDinconsistency} as
\begin{eqnarray}
 \eta &=& \frac{ \log \left( N_\omega \gamma_\omega \right)}{\log z_\omega} ,
 \label{eqn:consistencycutoff} \\
  \gamma_\omega  &\triangleq& \min_{n'} \frac{\min\{ p_i(n'), p_j(n') \}}{p^{(1-\omega)}_i(n')p^{\omega}_j(n')} \nonumber
\end{eqnarray}

As a result, one should expect estimation biases to become more severe for higher object numbers. For a small number of objects, these effects do not necessarily yield biases in MAP estimations, which  also explains the accurate estimates obtained using EMD fusion of C-PHD filters in simulated scenarios, e.g., in~\cite{Uney2013}. Next, we demonstrate this point in an example involving fusion of two binomial cardinality distributions.
\begin{example}
\label{sec:IIDClusterExample}
Let us consider the EMD fusion of two finite set distributions with binomial cardinalities given by $p_i(n) = B(n; k=5,P=0.95)$ and $p_j(n) = B(n; k=5,P=0.92)$ where these distributions give the probability that $n$ objects exist simultaneously among $k=5$ possibilities each with an existence probability of $P$ (\figurename~\ref{fig:iidinconsistency}\subref{fig:BB}). Of particular interest is the characteristics of $p_\omega$ as $z_\omega$ and $\omega$ vary in $0.1 \leq z_\omega \leq 0.9$ and $0\leq \omega \leq 1$, respectively. \figurename~\ref{fig:iidinconsistency}\subref{fig:BBfusion} presents fused distributions obtained by  varying $z_\omega$ and some intermediate values of $\omega$. Note that the cardinality $n$ at which the fused distributions peaks varies with $z_\omega$ as suggested in Corollary~\ref{cor:iid}. In particular, the inconsistency bound in \eqref{fig:calI} is illustrated in \figurename~\ref{fig:iidinconsistency}\subref{fig:calI} which monotonically increases with $n$ as discussed. The inconsistency threshold for $n$ as given in \eqref{eqn:consistencycutoff} is given in \figurename~\ref{fig:iidinconsistency}\subref{fig:eta}. Note that for a large ratio of $z_\omega$ and $\omega$ values, this threshold is larger than five and the MAP estimate for the cardinality given in \figurename~\ref{fig:iidinconsistency}\subref{fig:mapest1} agrees with the individual MAP estimates of $\hat n_i = \hat n_j =5$. However, there are also MAP estimates that indicate less than five objects caused by the IID inconsistency.
 
These computations are repeated for cardinality distributions peaking at a higher $n$ value. Specifically, $p_i(n) = B(n; k=35,P=0.98)$ and $p_j(n) = B(n; k=35,P=0.975)$ are used (see \figurename~\ref{fig:iidinconsistency}\subref{fig:cardHigh}) which have individual map estimates of $\hat n_i = \hat n_j = 35$. The fused cardinalities in \figurename~\ref{fig:iidinconsistency}\subref{fig:BBfusionHigh} illustrate that for a larger subset of $(z_\omega, \omega)$ pairs the IID inconsistency occurs, now, as discussed above. The resulting errors in estimating the number of objects is given in \figurename~\ref{fig:iidinconsistency}\subref{fig:mapBBHigh} which verifies our expectation: Based on that~\eqref{eqn:IIDinconsistency} approaches to $1$ with increasing $n$, as the peak cardinality increases, the IID inconsistency detoriates decision making more.
\end{example}

\subsection{Summary of results}
As a summary, this section has shown that when EMDs of finite set densities are used for their fusion, the resulting cardinality distribution will bear inconsistencies depending on $z_\omega(n)$. Proposition~\ref{prop:inconsistency} provides a general condition on the fused distribution to be inconsistent with the input distributions at a cardinality value $n$. This condition is specialised for Bernoulli finite set densities in Corollary~\ref{ref:corbernoulli}.  Example~\ref{ex:GaussBernoulli} has demonstrated that this condition holds with overwhelming probability for Bernoulli EMDs. In Poisson cardinality distributions, there is a single parameter $\lambda$ that specifies the distribution for all $n$. Proposition~\ref{prop:PoissonInconsistency} provides a condition of inconsistency in this parameter, similarly as an upper bound on $z_\omega(n)$. It is pointed out that because $z_\omega(n)$ is determined by the sensing diversity as well as sensor measurement histories in a sensor network, its value should be expected to be less than one in these settings\footnote{{The authors at this point would like to conjecture that $z_\omega(n)<1$ with probability one in a multi-sensor setting in which the finite set densities to be fused are posteriors obtained from recursive Bayesian filtering of local sensor data, i.e., $f_i(X)=f(X|Z^i_{1:t})$ and $f_j(X)=f(X|Z^j_{1:t})$ for realisations $Z^i_{1:t}$ and $Z^j_{1:t}$ of (independent) measurement processes associated with sensors $i$ and $j$, respectively.}} which in turn shows that Poisson EMDs are very prone to inconsistencies, as well. IID cluster processes have more general cardinality distributions. For the case, Proposition~\ref{prop:inconsistency} specialises to Corollary~\ref{cor:iid} which reveals that inconsistencies should be expected in MAP estimates of the cardinality, when the input densities indicate a high number of objects. These points are demonstrated in Example~\ref{sec:IIDClusterExample}.

\section{Cardinality consistent fusion of finite set distributions}
\label{sec:CardinalityConsistentFusion}
In this section, we propose a new approach that accommodates EMD fusion while avoiding the cardinality inconsistencies detailed in Section~\ref{sec:FiniteSetEMDsCardinalityDistributions}. These inconsistencies result from the dependency of the fused cardinality pmf on the scaling factor series $z_\omega(n)$. One way to remove this dependency is to decouple the fusion problem for different cardinalities by asserting a separate variational problem for each cardinality as opposed to using $\text{P2}$ in \eqref{eqn:centroid2} with finite set distributions as a single entity. 

\subsection{Variational problem definitions}
\label{sec:VarProblemDefinitions}
Let us first consider finite set distributions as parameterised in \eqref{eqn:RFSdensity} and remind that problem $\text{P2}$ is solved with distributions in the form given in~\eqref{eqn:emdloc}--\eqref{eqn:Nomega}. Now, let us consider the following family of variational problems given $f_i$ and $f_j$:
\begin{eqnarray}
\text{(P3)}\,\,\,\,\, {\text {For}}\,\,\, n&=&1,2,\ldots \nonumber\\
(\omega_n^*,\rho_{\omega_n^*,n}) &\triangleq& \arg \max_{\omega \in [0,1]} \min_{\rho_n \in {\cal P}_n } J_{\omega,n}[\rho_n]
  \label{eqn:centroidfamily} \\
 J_{\omega,n}[\rho_n] &\triangleq& (1-\omega)D(\rho_n || \rho_{i,n}) + \omega D(\rho_n || \rho_{j,n}).
 \nonumber
\end{eqnarray}

Here, ${\cal P}_n$ is the space of localisation densities with $n$ arguments which are symmetric in their arguments. Note that $\text{P3}$ is a set of $\text{P2}$ that has the localisation distributions for each cardinality $n$ as the entries, separately. Equivalently, $\text{P3}$ asserts the variational problem of fusion be treated as a conditional problem to be solved given $n$.

Following our discussion in Section~\ref{sec:problemdefinition}, solutions of these uncoupled problems have an EMD form given by~\eqref{eqn:emdloc}~and~\eqref{eqn:zomega}\footnote{It is easy to show that because $\rho_{i,n}$ and $\rho_{j,n}$ are symmetric in their arguments, $\rho_{\omega,n}$ also exhibits this symmetricity.}. One difference here compared to the solution of problem $\text{P2}$ is that for each $n$, a different optimal weight $\omega^*_n$ will be output, in general, as opposed to a single one. In addition --and, more importantly-- problem $\text{P2}$ decouples fusion of the cardinality distributions thus given $(\omega^*_n,\rho_{\omega^*_n,n})$, the fused cardinality distribution becomes an additional degree of freedom in the fused finite set distribution. In other words, the fusion of cardinality distributions can now be carried out in an isolated fashion in addition to problem $\text{P2}$ as a solution to 
\begin{eqnarray}
\label{eqn:p4}
\mspace{-30mu}\text{(P4)} \,\,\,\,\,\,(\omega_c^*,\tilde p_{\omega_c^*}) &\triangleq& \arg \max_{\omega \in [0,1]} \min_{ p \in {\cal P}_c } J_{\omega,c}[ p ]
   \\
  J_{\omega,c}[ p ] &\triangleq& (1-\omega)D( p || p_{i}) + \omega D( p || p_{j}).
\nonumber
\end{eqnarray}

Following the discussion in Section~\ref{sec:problemdefinition}, the solution to problem $\text{P4}$ is the EMD of the cardinality pmfs
\begin{eqnarray}
 \tilde p_\omega(n)&=& \frac{1}{\tilde N_\omega}p_i^{(1-\omega)}(n)p_j^{\omega}(n) 
 \label{eqn:tildecard} \\
 \tilde N_\omega &=& \sum_{n'=0} p_i^{(1-\omega)}(n')p_j^{\omega}(n').
 \label{eqn:tildeNomega}
\end{eqnarray}
evaluated at $\omega = \omega^*_c$.

This distribution differs from the cardinality of the solution to $\text{P2}$ (given in~\eqref{eqn:emdcard}~and~\eqref{eqn:Nomega}) in that it does not involve $z_\omega(n)$, and, is an EMD of the input finite set cardinalities. Therefore, the consistency condition (see~\eqref{eqn:consistency})
\begin{equation}
 \tilde p_{\omega}(n) \geq \min \{ p_i(n), p_j(n) \}
 \notag
\end{equation}
is satisfied for all $n$ and for all $\omega$ regardless of $z_\omega(n)$. Thus, $\tilde p_{\omega}$ prevents the decision errors stemming from the cardinality inconsistencies of the solutions to $\text{P2}$ as detailed in the previous section.

As a result, $\text{P3}$ and $\text{P4}$ yield a fused finite set density featuring cardinality consistency given by
\begin{eqnarray}
 \tilde f^*(X) &=&  \tilde f_{\Omega}(X) \bigg\rvert_{\Omega = (\omega_c = \omega^*_c,\omega_1 = \omega^*_1,\omega_2 = \omega^*_2,\ldots)}  \\
 \tilde f_{\Omega}(X) &\triangleq& p_{\omega_c}\left(  n \right) n! \rho_{\omega_n,n} (x_1,...,x_n)
 \label{eqn:tildefusion}
\end{eqnarray}
where $\omega^*_c$ is found by solving the maximisation in~\eqref{eqn:p4} with a cardinality distribution given by \eqref{eqn:tildecard}~and~\eqref{eqn:tildeNomega}. Here, $\omega^*_n$ solves the maximisation in~\eqref{eqn:centroidfamily} with a localisation distribution in~\eqref{eqn:emdloc}~and~\eqref{eqn:zomega}. These localisation distributions -- similar to the cardinality distribution-- are consistent individually, as they are EMDs of the inputs.

The pointwise consistency of $\tilde f$ over the space of finite sets, however, is not guaranteed. In order to clarify this point, we provide the following proposition:
\begin{proposition}[Pointwise inconsistency]
Let us consider
\begin{equation}
 \tilde f_{\omega}(X) \triangleq \tilde f_{\Omega}(X) \bigg\rvert_{\Omega = (\omega_c = \omega,\omega_1 = \omega,\omega_2 = \omega,\ldots)} 
\end{equation}
for some $\omega$. $\tilde f_\omega$ is pointwise inconsistent, i.e.,
\begin{equation}
 \tilde f_{\omega}(X) < \min \{ f_i(X), f_j(X) \}
 \label{eqn:PI}
\end{equation}
if
\begin{equation}
  \frac{ E_{\tilde p_\omega} \{ z_\omega(n) \}}{z_\omega(n) } < \frac{ \min \{ f_i(X), f_j(X) \} }{ f_\omega(X) }<=1
  \label{eqn:PIcondition}
\end{equation}
where $f_\omega$ is the finite set EMD given in~\eqref{eqn:emdloc}--\eqref{eqn:Nomega}. 
\end{proposition}
\begin{proof}
By comparing \eqref{eqn:emdloc}--\eqref{eqn:Nomega} and \eqref{eqn:tildecard}--\eqref{eqn:tildefusion}, it can be seen that the two finite set densities of concern are related by
\begin{equation}
\tilde f_\omega (X) = \frac{ E_{\tilde p_\omega} \{ z_\omega(n) \}}{z_\omega(n) } f_\omega(X).
\label{eqn:FSDrelation}
\end{equation}
where the expectation is with respect to~\eqref{eqn:tildecard}. 

Substitution of \eqref{eqn:FSDrelation} in~\eqref{eqn:PI} yields the first inequality in \eqref{eqn:PIcondition}. Note that, $f_\omega$ satistifies the pointwise consistency condition in \eqref{eqn:consistency}, hence, the right hand side of the inequality is smaller than or equal to one.
\end{proof}

This proposition points out that pointwise consistency of $\tilde f_\omega$ is guaranteed only for those $X$ with cardinality $n$ for which the scaling factor of the localisation density $z_\omega(n) \leq 1$ equals to the expectation. If $z_\omega(n)$ is greater than the expectation to the extent that \eqref{eqn:PIcondition} is satisfied, then $\tilde f_\omega$ exhibits pointwise inconsistency despite being consistent in the global cardinality and localisation distributions. As a conclusion, pointwise consistency does not imply consistency in global cardinality in fusion of finite set densities and vice versa.

\subsection{Solving the cardinality consistent fusion problems}
\label{sec:SolvingCardinalityConsistentFusion}
\begin{algorithm}[t]
\caption{Newton iterations for solving the $n$th problem in $\text{P3}$.} 
\label{Algorithm:NewtonIterationsP3}
\begin{algorithmic}[1]
\State {Input:} $\rho_{i,n}$,$\rho_{j,n}$ \Comment{Localisation densities} 
\State {Input:} $\omega^{(0)} \in [0,1]$, $\epsilon$ \Comment{Initial value, termination threshold}
\State $k \leftarrow 1$, $\omega^{(1)}\leftarrow \infty$
\While{$|\omega^{(k)}-\omega^{(k-1)}| > \epsilon$} \Comment{termination condition}
\State\label{P3:StepZ} $z_k \leftarrow z_\omega(n)\rvert_{\omega = {\omega^{(k)}} }$ using \eqref{eqn:zomega}
\State\label{P3:StepZdash} $z^\prime_k \leftarrow z^{\prime}_{\omega}(n) \rvert_{\omega = {\omega^{(k)}} }$ using \eqref{eqn:zprime}
\State\label{P3:StepZddash} $z^{\prime\prime}_k \leftarrow z^{\prime \prime}_{\omega}(n)\rvert_{\omega = {\omega^{(k)}} } $ using \eqref{eqn:zprimeprime}
\State\label{P3:NewtonUpdate} $\omega^{(k+1)} \leftarrow \omega^{(k)} - z^{\prime}_{ k } z_{ k } / \left( z^{\prime \prime}_{k} z_{k} - \left( z^{\prime}_{k} \right)^2 \right)$ 
\State $k \leftarrow k + 1$
\EndWhile
\State Return $\omega^*_n \leftarrow \omega^{(k)}$
\State Return $\rho_{\omega^*_n,n}$ using \eqref{eqn:emdloc}
\end{algorithmic}
\end{algorithm}
The variational problems $\text{P3}$ and $\text{P4}$ are $\max \min$ optimisation problems similar to \eqref{eqn:centroid2}. Hence, the minimisations given $\omega$ are solved by the EMDs of their argument distributions (see the discussion in Section~\ref{sec:problemdefinition} and Appendix~\ref{sec:statpointsproblemp}). The objective of the outer maximisation in $\text{P3}$ is therefore (see also~\eqref{eqn:cost3})
\begin{eqnarray}
 G_n (\omega_n) &\triangleq& J_{\omega,n }[ \rho_n ]  \bigg\rvert_{\rho_n = \rho_{\omega,n}} \label{eqn:Gn} \\
 &=& -(\omega -1){\cal R}_{\omega}(\rho_{n,i},\rho_{n,j}) \nonumber \\
 &=& - \log z_\omega(n) \nonumber
\end{eqnarray}
which is a concave function of its one dimensional argument that takes values from a bounded interval. Newton iterations converge to a solution and have an excellent convergence rate~\cite{Bazaraa1993}. Starting from an initial value $\omega^{(0)} \in [0,1]$, recursive increments are made by the ratio of the first and second order derivatives, i.e.,
\begin{eqnarray}
\mspace{-20mu} \frac{ G^{\prime}_n(\omega_n) }{ G^{\prime\prime}_n(\omega_n) }  \bigg\rvert_{ {\omega_n = \omega^{(k)}} }\mspace{-10mu} &=&  \frac{ z^{\prime}_{ \omega }(n) z_{ \omega }(n) }{ z^{\prime \prime}_{\omega}(n) z_{\omega}(n) - \left( z^{\prime}_{\omega }(n) \right)^2  } \bigg\rvert_{ {\omega = \omega^{(k)}} } \\
 z^{\prime}_{ \omega }(n) &\triangleq&\int\rho_{i,n}^{1-\omega}(x)\rho_{j,n}^{\omega }(x) \log \frac{\rho_{j,n}(x)}{\rho_{i,n}(x) }\mathrm{d}x \label{eqn:zprime}\\
 z^{\prime\prime}_{ \omega }(n) &\triangleq&\int\rho_{i,n}^{1-\omega}(x)\rho_{j,n}^{\omega }(x) \left( \log \frac{\rho_{j,n}(x)}{\rho_{i,n}(x) } \right)^2\mathrm{d}x  \nonumber \\[-12pt]\label{eqn:zprimeprime}
\end{eqnarray}
where $\omega^{(k)}$ is the value found in the $k$th iteration. Here, $z_\omega(n)$ is given in~\eqref{eqn:zomega} and its derivatives in~\eqref{eqn:zprime}~and~\eqref{eqn:zprimeprime} are found in Appendix~\ref{sec:DerivationAlg1}. Algorithm~\ref{Algorithm:NewtonIterationsP3} explicitly specifies this iterative solution which takes the localisation densities as inputs together with an initial value and termination condition. Upon convergence, the optimal value $\omega_n^* = \omega_n^{(k)}$ is found for which the corresponding EMD $\rho_{\omega_n^{(k)},n}$ is the fused density.

An analogous iterative algorithm for finding the consistently fused cardinality distribution as a solution to $\text{P4}$ is given in Algorithm~\ref{Algorithm:NewtonIterationsP4}. Note that the computations involved here can be carried out exactly for distributions with finite support, in practice.
 
\begin{algorithm}[t]
\caption{Newton iterations solving $\text{P4}$ for consistent cardinality fusion.} 
\label{Algorithm:NewtonIterationsP4}
\begin{algorithmic}[1]
\State {Input:} $p_i(n)$,$p_j(n)$ \Comment{Cardinality pmfs} 
\State {Input:} $\omega^{(0)} \in [0,1]$, $\epsilon$ \Comment{Initial value, termination threshold}
\State $k \leftarrow 1$, $\omega^{(1)}\leftarrow \infty$
\While{$|\omega^{(k)}-\omega^{(k-1)}| > \epsilon$} \Comment{termination condition}
\State $N_k \leftarrow N_\omega \rvert_{\omega = {\omega^{(k)}} }$ using \eqref{eqn:Nomega}
\State $N^\prime_k \leftarrow \sum p_i^{1-\omega}(n)p_j^{\omega}(n) \log \big( p_j(n)/p_i(n) \big)$ 
\State $N^{\prime\prime}_k \leftarrow \sum p_i^{1-\omega}(n)p_j^{\omega}(n) \bigg( \log \big( p_j(n)/p_i(n) \big) \bigg)^2$
\State $\omega^{(k+1)} \leftarrow \omega^{(k)} - N^{\prime}_{ k } z_{ k } / \left( N^{\prime \prime}_{k} N_{k} - \left( N^{\prime}_{k} \right)^2 \right)^2$ 
\State $k \leftarrow k + 1$
\EndWhile
\State Return $\omega^*_c \leftarrow \omega^{(k)}$
\State Return $p_{\omega^*_c}$ using \eqref{eqn:emdcard}
\end{algorithmic}
\end{algorithm}
 
Algorithm~\ref{Algorithm:NewtonIterationsP3}, on the other hand, should accommodate adequate computational schemes for exactly or approximately evaluating the integrals involved. In the latter case, it admits the interpretation of being a stochastic gradient approach~\cite{Bubeck2015}. Specification of such procedures is beyond the scope of this work.

There are, nevertheless, structural simplifications in both $\text{P3}$ and $\text{P4}$ for different families of finite set families. For Poisson and IID cluster finite set distributions, the localisation densities are parameterised by a single density over a single state variable --as given in \eqref{eqn:Poissonloc}-- which is the same for different cardinalities. In addition, the solution to the inner minimisation in $\text{P3}$ (equivalently $\text{P}$ in \eqref{eqn:centroid}) is also a Poisson and IID cluster, respectively, in the Poisson and IID cluster cases~\cite{Uney2013}. Thus, $D(\rho_n || \rho_{n,i} )=n D(\rho || \rho_i )$ where $\rho$ parameterises $\rho_n$, and, the family of problems in $\text{P3}$ satisfy
\begin{equation}
J_{\omega,n}[\rho_n] = n J_{\omega,1}[\rho].
\end{equation}
Consequently, the optimal solution to $\text{P3}$ for cardinality $n$ is parameterised by the optimal solution to $n=1$ thereby restricting it to the case for only $n=1$. 

Bernoulli finite sets have nonzero cardinality pmf only for $n \leq 1$ naturally restricting $\text{P3}$ to $n=1$. If the parameterising densities $\rho_i$ and $\rho_j$ are Gaussians, then $\text{P3}$ specifies a covariance intersection procedure~\cite{Hurley2002}. For this case, $\text{P4}$ has a closed form solution given by~\cite{Nielsen2011a}
\begin{eqnarray}
  \omega^*_c &=& \frac{ \log \bigg( \frac{\log (1-\alpha_i)/(1-\alpha_j) }{\log (\alpha_j/\alpha_i) } \bigg) - \log \bigg( \frac{\alpha_i}{1-\alpha_i} \bigg) }{ \log \frac{1-\alpha_i}{1-\alpha_j} + \log \frac{\alpha_j}{\alpha_i} } \label{eqn:beroptomega} \\
 \alpha^* &= &\frac{ \alpha_i^{1-\omega^*_c } \alpha_j^{\omega^*_c} }{\alpha_i^{1-\omega^*_c } \alpha_j^{\omega^*_c} + (1-\alpha_i)^{1-\omega^*_c}(1-\alpha_j)^{\omega^*_c} } \label{eqn:beroptalpha}
\end{eqnarray}

For Poisson cardinality pmfs, similarly a closed form solution exists for $\text{P4}$ which is given by~\cite{Nielsen2011a}
\begin{eqnarray}
 \omega^*_c &=& \frac{ -\log\big( \log (\lambda_j/\lambda_i) \big) + \log (\lambda_j/\lambda_i -1) }{ \log (\lambda_j/\lambda_i) }  \label{eqn:poisoptomega} \\
 \lambda^* &=& \lambda_i^{1-\omega^*_c}\lambda_j^{\omega^*_c} \label{eqn:poisoptlambda}
\end{eqnarray}

It is worthwhile to notice that both Bernoulli and Poisson distributions are exponential family distributions and the above solutions bear the geometric properties aforementioned in Section~\ref{sec:genci} and proved in~\cite{Nielsen2013}.

\subsection{Demonstration of cardinality consistent fusion}
 
In this section, we revisit the examples in Section~\ref{sec:FiniteSetEMDsCardinalityDistributions} involving cardinality inconsistencies and demonstrate the efficacy of the solutions of Problems $\text{P3}$ and $\text{P4}$ in these fusion scenarios. 
 
\begin{example}[Gauss-Bernoulli case revisited]
\label{ex:GaussBernoulliRe}
Let us consider the Gauss-Bernoulli case in Example~\ref{ex:GaussBernoulli} in the light of the discussion above. Fusion of localisation distributions in $\text{P3}$ involve the fusion of only a single pair for $n=1$, for the case. These distributions $\rho_i$ and $\rho_j$ given by \eqref{eqn:GaussianLocs}  are Gaussians, therefore, Algorithm~\ref{Algorithm:NewtonIterationsP3} is equivalently an iterative covariance intersection algorithm that optimises $\omega$ to achieve the KLD equality criteria in~\eqref{eqn:StationaryOmega}. The Newton update for the parameter $\omega$ in Step~\ref{P3:NewtonUpdate} of Algorithm~\ref{Algorithm:NewtonIterationsP3} is carried out as follows: Evaluation of $z_\omega$ in Step~\ref{P3:StepZ} at $\omega = \omega^{(k)}$ is made using its closed form expression in~\eqref{eqn:GaussEMDscale}. Given $z_\omega$, the derivative in Step~\ref{P3:StepZdash} is also found in closed form using the following identity
\begin{equation}
z^{\prime}_{\omega} = z_\omega \big( D(\rho_{\omega}||\rho_i) - D(\rho_{\omega}||\rho_j) \big),
\end{equation}
which can easily be verified by dividing both sides of \eqref{eqn:zprime} to $z_\omega(n)$. This quantity is computed by evaluating the KLD of multi-variate Gaussian densities given by (see, for example,~\cite[A.23]{Rasmussen2006})
\begin{multline}
 D(\rho_{\omega}||\rho_i) = {1/2}\log | \vect{C}_i \vect{C}_{\omega}^{-1} | \\ 
 + \vect{tr}\{ \vect{C}_i^{-1}\big((\vect{m}_\omega-\vect{m}_i)(\vect{m}_\omega-\vect{m}_i )^T + \vect{C}_\omega - \vect{C}_i \big)\},
\end{multline}
where $ \vect{tr}\{.\}$ is the trace of its matrix argument. The second derivative in Step~\ref{P3:StepZddash} on the other hand, is found approximately using the Monte Carlo method~\cite[Chp.3]{Robert2004} targeting the integration in \eqref{eqn:zprimeprime} divided by $z_\omega(n)$. Therefore, $L$ samples are generated from $\rho_\omega$ for this step, i.e., $x^{(l)} \sim \rho_\omega$ for $l=1,
\ldots,L$. Using these samples, the approximation seeked is given by
\begin{equation}
z^{\prime\prime}_\omega \approx {z_\omega} \times \frac{1}{L}\sum_{l=1}^L \Big( \log \frac{\rho_j(x^{(l)})}{\rho_i(x^{(l)})} \Big)^2
\label{eq:}
\end{equation}
where the approximation error decreases with ${\cal O}(1/\sqrt{L})$. 

We input the Gaussian pairs in Example~\ref{ex:GaussBernoulli} that are obtained by varying the covariance condition number $\kappa$ -- equivalently the sensing diversity-- to Algorithm~\ref{Algorithm:NewtonIterationsP3} and use the computational procedures above. \figurename~\ref{fig:gaussSGCIs} depicts optimally weighted EMDs for $\kappa=1,10,20$. The optimal weight $\omega^*$ output by Algorithm~\ref{Algorithm:NewtonIterationsP3} as a function of the condition number $\kappa$ is given in \figurename~\ref{fig:kappa_vs_omega}. Here, the termination threshold is set to $\epsilon = 1\mathrm{e}{-4}$ and the number of samples used for the Monte Carlo estimate in Step~\ref{P3:StepZddash} is $L=1000$. Convergence is declared after an average of $3.4$ and a maximum of $5$ iterations. Note that $\omega^*$ yields a fused result that bears more influence from $\rho_i$ as $\kappa$ increases.

\begin{figure}[bt!]
  \centering{
  \begin{minipage}{\linewidth}
    \includegraphics[width=0.31\linewidth]{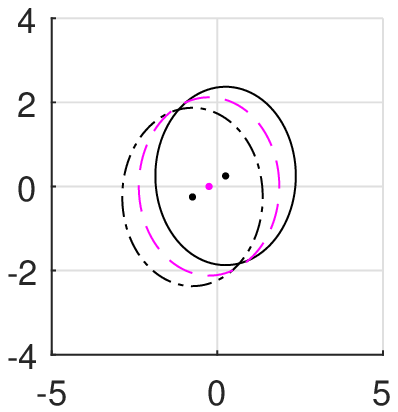}
    \includegraphics[width=0.31\linewidth]{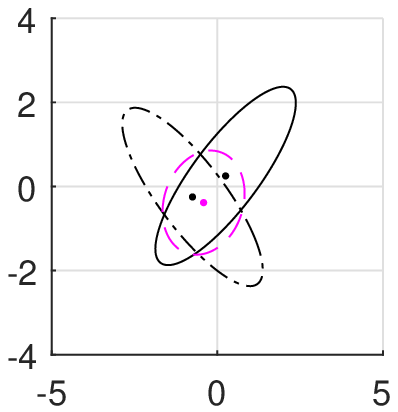}
    \hfill
    \includegraphics[width=0.31\linewidth]{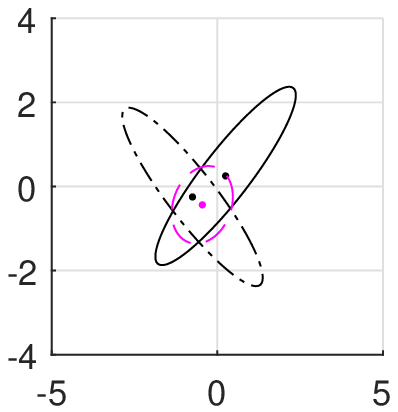}
  \end{minipage}
  \vspace{-4pt}
  }
  \caption[Optimal EMDs]{Optimally weighted EMDs of Gaussian localisation distributions obtained by using Algorithm~\ref{Algorithm:NewtonIterationsP3}: Here, $\rho_i$ (solid line)~and~$\rho_j$ (dash-dotted line) are inputs for $\kappa=1,10,20$ (left to right -- see Example~\ref{ex:GaussBernoulli} for details)~and~$\rho_{\omega^*}$ (magenta dashed line) are fused outputs with optimal weights found as $\omega^*= 0.500, 0.397$ and $0.387$ (left to right).}
  \label{fig:gaussSGCIs}
\end{figure}

\begin{figure}[t!]
  \centering{
  \begin{minipage}{\linewidth}
    \includegraphics[width=0.9\linewidth]{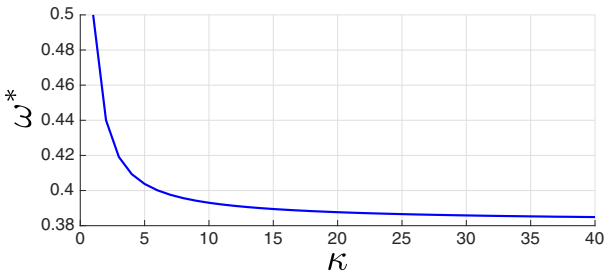}
  \end{minipage}
  \vspace{-1pt}
  }
  \caption[Optimal omega values]{Optimal weight parameters found using Algorithm~\ref{Algorithm:NewtonIterationsP3} as a function of the sensing diversity $\kappa$ as explained in detail in Example~\ref{ex:GaussBernoulli}.}
  \label{fig:kappa_vs_omega}
\end{figure}

\begin{figure}[b!]
  \centering{
  \begin{minipage}{\linewidth}
    \includegraphics[width=0.9\linewidth]{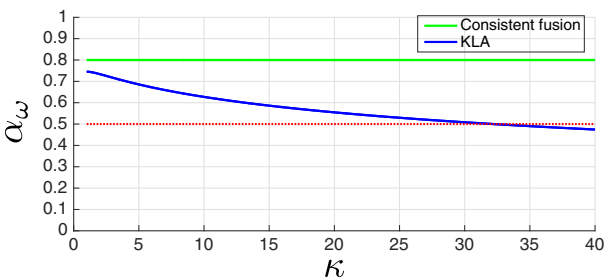}
  \end{minipage}
  \vspace{-4pt}
  }
  \caption[Fused existence probabilities]{Fused existence probabilities: Cardinality consistent fusion via Problem $\text{P4}$ (green line) in comparison with fused existence probabilities obtained using KL averaging (blue line--see, e.g.,\cite{Battistelli2013}--\cite{Li2018}) depicted as a function of sensing diversity parameter $\kappa$ as detailed in Example~\ref{ex:GaussBernoulli}. The red-dashed line is the canonical Bayesian decision threshold for deciding the existence of an object.}
  \label{fig:kappa_vs_alpha}
\end{figure}

Cardinality fusion problem $\text{P4}$ has an analytical solution for the case. Using \eqref{eqn:beroptomega} for $\alpha_i=\alpha_j=0.8$ (see Example~\ref{ex:GaussBernoulli}), we find that $\omega_C^*=0.5$ and $\alpha^*=0.8$ regardless of the solution of Problem $\text{P3}$ above, i.e., the optimal weight parameters $\omega^*$ of the localisation distributions or the corresponding normalisation constants $z_{\omega^*}$s. 

Let us compare this result with conventional EMD fusion with weights selected using ``Kullback-Leibler averaging'' (KL averaging) as used in, for example, \cite{Battistelli2013}--\cite{Li2018}. In this approach, the fused finite set density is the Bernoulli EMD with weight parameter $\omega = 1/2$. In other words, the fused result solves Problem $\text{P}$ for $\omega = 1/2$. The coupling of the fused existence probability with $z_\omega$ and the lack of a weight selection mechanism results with the fused existence probabilities given in \figurename~\ref{fig:kappa_vs_alpha} which illustrates a monotonically increasing disagreement with the input beliefs on the existence of an object as the sensing diversity  increases. This trend results with the fused existence probability falling below the canonical decision threshold of~$0.5$ \footnote{Note that this graph is nothing but the cross-section of the existence probability graph in \figurename~\ref{fig:gbcombo} along $\omega=0.5$.}. The proposed cardinality consistent fusion, on the other hand, preserves the confidence of input distributions on the existence of an object irrespective of the sensing diversity. The KL averaging fusion outputs a localisation density that is similarly the EMD of the localisation distributions with weight $\omega = 1/2$. This margin between this value and the optimum point found by the proposed algorithm grows significantly with $\kappa$, in this example (see \figurename~\ref{fig:kappa_vs_omega}).

Note that these results are also relevant for the work in literature on fusion of multi-Bernoulli~\cite{Yi2017}, and, labelled random finite set families as their fusion is often reduced to performing Bernoulli-Bernoulli fusion for multiple pairs using, for example, KL averaging~\cite{Li2018}.

\end{example}

\begin{example}[Example~\ref{sec:IIDClusterExample} revisited]
Let us demonstrate Algorithm~\ref{Algorithm:NewtonIterationsP4} in solving the cardinality fusion problem $\text{P4}$. First, we consider the binomial cardinality distribution pair illustrated in \figurename~\ref{fig:iidinconsistency}\subref{fig:BB}. Note that, the algorithm allows for exact computations in all steps. The termination threshold is selected as $\epsilon=1.0e-4$. In $2$ iterations Algorithm~\ref{Algorithm:NewtonIterationsP4} declares convergence to the optimal weight parameter $\omega_C^*= 0.5182$. The corresponding fused cardinality pmf is the EMD with this weight and depicted in \figurename~\ref{fig:cf}\subref{fig:cfA}. We repeat the same procedure for the cardinality pair in \figurename~\ref{fig:iidinconsistency}\subref{fig:BBfusionHigh}. The proposed algorithm convergences in $2$ steps to $\omega_C^*=0.5090$. The resulting cardinality distribution is depicted in \figurename~\ref{fig:cf}\subref{fig:cfB}. Note that the MAP estimates of the number of objects is in agreement with the inputs. The consistency here is underpinned by that the cardinality fusion here is independent of localisation densities and~$z_\omega$.

\begin{figure}[t!]
\vspace{-5pt}
  \centering{
  \subfloat[]{ \includegraphics[width=0.47\linewidth,height=37mm]{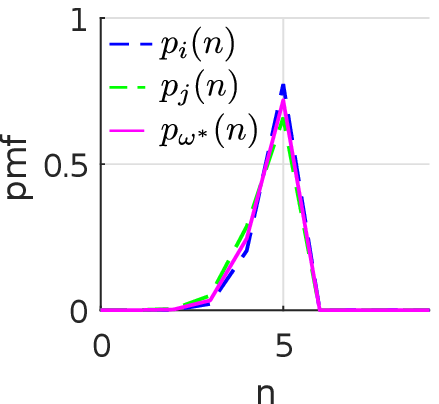} \label{fig:cfA} }\hfill
  \subfloat[]{ \includegraphics[width=0.47\linewidth,height=37mm]{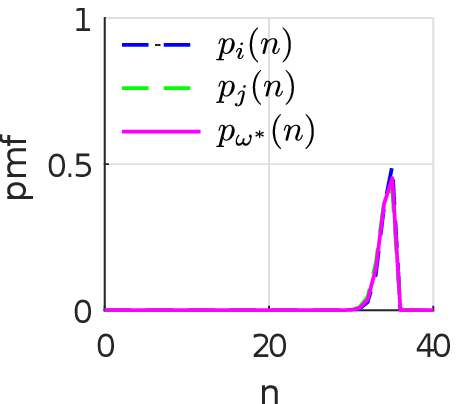} \label{fig:cfB} }
  }
  \caption[Consistent fusion of binomails]{Consistent cardinality fusion using Algorithm~\ref{Algorithm:NewtonIterationsP4}: Inputs (blue and green dashed lines) are the binomial pairs introduced in Example~\ref{ex:GaussBernoulli}. The optimal weights converged are $\omega_C^*= 0.5182$~and~$0.5090$ in \subref{fig:cfA}~and~\subref{fig:cfB}, respectively.}
  \label{fig:cf}
\end{figure}

In order to contrast this result with that obtained by KL averaging (see, e.g.,\cite{Battistelli2013}), let us consider the fused cardinalities depicted for $\omega=1/2$ in the top right panes in \figurename~\ref{fig:iidinconsistency}\subref{fig:BBfusion}~and~\subref{fig:BBfusionHigh} in which the coupling of cardinality fusion with $z_\omega$ is demonstrated. Let us remind also that the MAP object number estimate depends on $z_\omega$ (\figurename{s}~\ref{fig:iidinconsistency}\subref{fig:mapest1}~and~\subref{fig:mapBBHigh}) with an increasing bias towards underestimation as the input distribution peaks shift towards right indicating higher number of objects (\figurename{s}~\ref{fig:iidinconsistency}\subref{fig:mapBBHigh}). The proposed algorithm, on the other hand, finds a consistent common ground of the input cardinality distributions which is also optimal with respect to Problem $\text{P4}$.
\end{example}

\section{Conclusion}
\label{sec:Conclusion}
This work considered the recently growing literature on the use of EMDs --or, weighted geometric means- of finite set densities in multi-sensor fusion for multi-object tracking. EMDs of distributions are pointwise consistent, however, we have proved in this article that they are prone to inconsistency in their cardinality distributions which can lead to serious decision errors related to the number of objects sensed. We have demonstrated that pointwise consistency does not imply consistency in cardinality and vice versa. We remedy this problem by redefining the variational optimisation problem that underlies EMD fusion. Then, we specify iterative solutions and establish a conceptual framework for cardinality consistent fusion of finite set densities which also accommodates EMDs.

Following these results, possible future directions include investigation of numerical computational schemes in order to use within this variational framework. The extension of this variational perspective to accommodate $N$ sources is also anticipated to be a worthwhile direction to pursue. 

\appendix
\subsection{Solution of the Problem~$\text{P}$ in \eqref{eqn:centroid}}
\label{sec:statpointsproblemp}
In this appendix, we provide a direct proof for the assertion that the solution of Problem~$\text{P}$ in~\eqref{eqn:centroid} for any $\omega$ is the EMD given in~\eqref{eqn:emd},\eqref{eqn:emdscale} when $f$ is constrained to be a density, i.e., to integrate to unity. This constraint on the feasible set of solutions together with the cost functional of the problem are captured in the Lagrangian given by~\cite{Bazaraa2006}
\begin{eqnarray}
 {\cal L}[f,\lambda] &\triangleq& J_\omega[f] + \lambda G[f],
 \label{eqn:Lagrangian} \\ 
 G[f] &\triangleq& \left( 1-  \int f(X) \mathrm{d} X\right). \nonumber
\end{eqnarray}
Here,  $\lambda$ is a free variable referred to as a Lagrange multiplier which --at its stationary point-- imposes the constraint of integration to unity on those $f$ which are also stationary. This point together with the convexity of $J_\omega$ in $f$ results with the stationary point of \eqref{eqn:Lagrangian} $f^*$ being the solution to $\text{P}$ in~\eqref{eqn:centroid}~\cite{Bazaraa2006}.

The necessary (and sufficient) condition of stationarity that $f^*$ should satisfy is given by i) the functional derivative of the Lagrangian, i.e.,
\begin{equation}
 \bigg. \delta \left( J_\omega + \lambda G \right)[f; \delta_X ] \bigg\rvert_{f=f^*} = 0
 \label{eqn:condition1}
\end{equation}
for all $X$, and, ii) the partial differential with respect to the multiplier $\lambda$, i.e.,
\begin{equation}
  \bigg. \frac{\partial \left( J_\omega[f] + \lambda G[f] \right) }{\partial \lambda}  \bigg\rvert_{\lambda=\lambda^*} = 0.
 \label{eqn:condition2}
\end{equation}

The functional derivative in~\eqref{eqn:condition1} can be expressed in terms of the partial derivative with respect to $f(X)$, i.e.,
\begin{equation}
\delta \left( J_\omega + \lambda G \right)[f; \delta_X ] = \frac{\partial \left( J_\omega[f] + \lambda G[f] \right)}{\partial f(X)}  \notag
\end{equation}
for all $X$. By using the definition of KLD in~\eqref{eqn:KLD}, and rules of differentiation, this expression leads to 
\begin{multline}
 \frac{\partial \left( J_\omega[f] + \lambda G[f] \right)}{\partial f(X)} = \log f(X) + \frac{f(X)}{f(X)} \\ - (1-\omega) \log f_i(X) - \omega \log f_j(X) -\lambda.
 \nonumber
\end{multline}
The equation above is zero when the multiplier $\lambda$ takes the value
\begin{equation}
 \lambda = 1 + \log \frac{f(X)}{f^{(1-\omega)}_i(X)f^{\omega}_j(X) },
\end{equation}
for which the corresponding $f(X)$ is found as
\begin{equation}
 f(X) = \exp(\lambda-1)f^{(1-\omega)}_i(X)f^{\omega}_j(X).
 \label{eqn:foflambda}
\end{equation}

We substitute from the equality above into the partial differentiation in \eqref{eqn:condition2} and obtain
\begin{eqnarray}
&& \mspace{-50mu} \frac{\partial \left( J_\omega[f] + \lambda G[f] \right) }{\partial \lambda} \nonumber \\ 
&=&\frac{\partial}{\partial \lambda} \left( (\lambda-1)\exp(\lambda-1)\int f^{(1-\omega)}_i(X)f^{\omega}_j(X) \mathrm{d}X  \right) \notag\\
& &\,\,\,\,+ \frac{\partial}{\partial \lambda} \left(\lambda - \lambda \exp(\lambda-1) \int f^{(1-\omega)}_i(X)f^{\omega}_j(X) \mathrm{d}X   \right) \notag
\end{eqnarray}

\begin{eqnarray}
&=& \frac{\partial}{\partial \lambda} \left(\lambda -\exp(\lambda-1) \int f^{(1-\omega)}_i(X)f^{\omega}_j(X) \mathrm{d}X   \right) \notag \\
&=&1- \exp(\lambda-1)\int f^{(1-\omega)}_i(X)f^{\omega}_j(X) \mathrm{d}X.
\label{eqn:condition2b}
\end{eqnarray}

After \eqref{eqn:condition2b} is set to zero, $\lambda^*$ is found as
\begin{equation}
 \lambda^* = - \log \int f^{(1-\omega)}_i(X)f^{\omega}_j(X) \mathrm{d}X + 1
 \label{eqn:optmult}
\end{equation}
and, the solution $f^*$ is found by subsituting from \eqref{eqn:optmult} into \eqref{eqn:foflambda}~as
\begin{equation}
 f^*(X) = \exp\bigg( - \log{\int f^{(1-\omega)}_i(X)f^{\omega}_j(X)\mathrm{d}X} \bigg) f^{(1-\omega)}_i(X)f^{\omega}_j(X). \notag
\end{equation}
Following a rearrangement of the terms, the expression above takes the form given~by
\begin{multline}
 f^*(X)= \frac{1}{\int f^{(1-\omega)}_i(X)f^{\omega}_j(X)\mathrm{d}X}f^{(1-\omega)}_i(X)f^{\omega}_j(X),
 \label{eqn:P2solution}
\end{multline}
which can be identified as the EMD in~\eqref{eqn:emd},\eqref{eqn:emdscale}.
\vspace{-10pt}
\subsection{Proof of Proposition~\ref{prop:inconsistency}}
\label{sec:proof}
The proof follows from decomposing $p_\omega(n)$ in~\eqref{eqn:emdcard}~and~\eqref{eqn:Nomega} as follows
\begin{multline}
p_\omega(n) = \\
\frac{p_i^{(1-\omega)}(n)p_j^{\omega}(n)z_\omega(n)}{p_i^{(1-\omega)}(n)p_j^{\omega}(n)z_\omega(n)+ \sum_{n' \neq n} p_i^{(1-\omega)}(n')p_j^{\omega}(n') z_\omega(n') },
\end{multline}
and substituting on the left hand side of the inequality in~\eqref{eqn:inconsistency}. The inequality can easily be solved for $z_\omega(n)$ leading to~\eqref{eqn:condition}.

\subsection{Derivation of Algorithm~\ref{Algorithm:NewtonIterationsP4}}
\label{sec:DerivationAlg1}
We start by finding the first and second order derivatives of $G_n$ in~\eqref{eqn:Gn}:
\begin{eqnarray}
G^\prime_n(\omega_n) &\triangleq&  \frac{d G_n(\omega)}{d \omega} \bigg\rvert_{\omega = \omega_n} \nonumber \\
   &=& -\frac{ dz_{\omega}(n)/d\omega }{ z_{\omega}(n) } \bigg\rvert_{\omega = \omega_n}, 
 \label{eqn:Gndash}\\
 G^{\prime\prime}_n(\omega_n) &\triangleq&  \frac{d^2 G_n(\omega)}{d \omega^2}\bigg\rvert_{\omega = \omega_n} \nonumber \\
 &=& -\frac{   d^2z_{\omega}(n)/d\omega^2 \times z_\omega(n) - \left( dz_{\omega}(n)/d\omega \right)^2 }{ \left( z_\omega(n) \right)^2}\bigg\rvert_{\omega = \omega_n}.
 \label{eqn:Gnddash}
\end{eqnarray}
where $z_\omega(n)$ is given by~\eqref{eqn:zomega}. Newton iterations~\cite{Bazaraa1993} use recursive increments to the scalar argument of maximisation. These increments are found by evaluating the ratio of \eqref{eqn:Gndash} and~\eqref{eqn:Gnddash} which is found as
\begin{eqnarray}
\frac{ G^\prime_n(\omega_n)}{ G^{\prime\prime}_n(\omega_n) } &=& \frac{  z^\prime_\omega(n) z_\omega(n) }{  z^{\prime\prime}_\omega(n) z_\omega(n) - \left( z^\prime_\omega(n) \right)^2 }\bigg\rvert_{\omega = \omega_n}, \\
z^\prime_\omega(n)&\triangleq&   \frac{d z_\omega(n) }{ d\omega} \label{eqn:zomegaP}\\
z^{\prime\prime}_\omega(n)&\triangleq& \frac{ d z^\prime_\omega(n) }{ d \omega} \label{eqn:zomegePP}
\end{eqnarray}
in terms of $z_\omega(n)$ and its first and second order derivatives. 

Next, let us find the derivatives of $z_\omega(n)$. The first order derivative follows after substituting \eqref{eqn:zomega} in \eqref{eqn:zomegaP} as
\begin{eqnarray}
 \frac{d z_\omega(n) }{ d\omega} &= &\int \rho_{i,n}(x) \frac{d}{d\omega}\left( \frac{\rho_{j,n}(x)}{\rho_{i,n}(x)} \right)^\omega \mathrm{d}x \nonumber \\
 &=& \int \rho^{1-\omega}_{i,n}(x)\rho^\omega_{j,n}(x) \log \frac{\rho_{j,n}(x)}{\rho_{i,n}(x)} \mathrm{d}x \nonumber
\end{eqnarray}
which is equivalent to \eqref{eqn:zprime}. The second order derivative seeked is found by substituting from the above equality into \eqref{eqn:zomegePP} as 
\begin{eqnarray}
 \frac{ d z^\prime_\omega(n) }{ d \omega} &=& \int \rho_{i,n}(x) \log \frac{\rho_{j,n}(x)}{\rho_{i,n}(x)}\frac{d}{d\omega}\left( \frac{\rho_{j,n}(x)}{\rho_{i,n}(x)} \right)^\omega \mathrm{d}x \nonumber \\
& =&\int \rho^{1-\omega}_{i,n}(x)\rho^\omega_{j,n}(x) \bigg( \log \frac{\rho_{j,n}(x)}{\rho_{i,n}(x)} \bigg)^2 \mathrm{d}x\nonumber
\end{eqnarray}
which is equivalently given in \eqref{eqn:zprimeprime}.



%

\ifCLASSOPTIONcaptionsoff
 \newpage
\fi

\bibliographystyle{IEEEtran}

\begin{thebibliography}{10}
\providecommand{\url}[1]{#1}
\csname url@samestyle\endcsname
\providecommand{\newblock}{\relax}
\providecommand{\bibinfo}[2]{#2}
\providecommand{\BIBentrySTDinterwordspacing}{\spaceskip=0pt\relax}
\providecommand{\BIBentryALTinterwordstretchfactor}{4}
\providecommand{\BIBentryALTinterwordspacing}{\spaceskip=\fontdimen2\font plus
\BIBentryALTinterwordstretchfactor\fontdimen3\font minus
  \fontdimen4\font\relax}
\providecommand{\BIBforeignlanguage}[2]{{%
\expandafter\ifx\csname l@#1\endcsname\relax
\typeout{** WARNING: IEEEtran.bst: No hyphenation pattern has been}%
\typeout{** loaded for the language `#1'. Using the pattern for}%
\typeout{** the default language instead.}%
\else
\language=\csname l@#1\endcsname
\fi
#2}}
\providecommand{\BIBdecl}{\relax}
\BIBdecl

\bibitem{Hall2013}
D.~Hall, C.-Y. Chong, J.~Llinas, and M.~L. {II}, Eds., \emph{Distributed Data
  Fusion for Network-Centric Operations}.\hskip 1em plus 0.5em minus
  0.4em\relax CRC Press, 2013.

\bibitem{Julier2006a}
S.~Julier, T.Bailey, and J.~Uhlmann, ``Using exponential mixture models for
  suboptimal distributed data fusion,'' in \emph{Proc. of the 2006 IEEE
  Nonlinear Stat. Signal Proc. Workshop (NSSPW'06)}.\hskip 1em plus 0.5em minus
  0.4em\relax NSSPW'06, September 2006, pp. 160--163.

\bibitem{inproceedings:julier97c}
S.~J. Julier{~and J. K. Uhlmann}, ``{A Non-divergent Estimation Algorithm in
  the Presence of Unknown Correlations},'' in \emph{Proceedings of the IEEE
  American Control Conference}, vol.~4, Albuquerque NM, USA, June 1997, pp.
  2369--2373.

\bibitem{Hurley2002}
M.~Hurley, ``An information-theoretic justification for covariance intersection
  and its generalization,'' in \emph{Proceedings of the 2002 FUSION
  Conference}.\hskip 1em plus 0.5em minus 0.4em\relax FUSION 2002, July 2002.

\bibitem{Cover2006}
T.~M. Cover and J.~A. Thomas, \emph{Elements of Information Theory}, Second,
  Ed.\hskip 1em plus 0.5em minus 0.4em\relax John Wiley and Sons, 2006.

\bibitem{Heskes1998}
T.~Heskes, ``Selecting weighting factors in logarithmic opinion pools,'' in
  \emph{Advances in Neural Information Processing Systems}.\hskip 1em plus
  0.5em minus 0.4em\relax The MIT Press, 1998, pp. 266--272.

\bibitem{Mahler2000a}
R.~Mahler, ``Optimal/robust distributed data fusion: a unified approach,'' in
  \emph{Proceedings of the SPIE Defense and Security Symposium 2000}.\hskip 1em
  plus 0.5em minus 0.4em\relax SPIE Defense and Security Symposium, 2000.

\bibitem{Mahler2007}
R.~P.~S. Mahler, \emph{Statistical Multisource Multitarget Information
  Fusion}.\hskip 1em plus 0.5em minus 0.4em\relax Springer, 2007.

\bibitem{Mahler2003}
R.~Mahler, ``Multi-target {B}ayes filtering via first-order multi-target
  moments,'' \emph{IEEE Transactions on Aerospace and Electronic Sysmtes},
  vol.~39, no.~4, pp. 1152--1178, October 2003.

\bibitem{Vo2006}
B.~Vo and W.~K. Ma, ``The {G}aussian mixture probability hypothesis density
  filter,'' \emph{IEEE Transactions on Signal Processing}, vol.~54, no.~11, pp.
  4091--4104, November 2006.

\bibitem{Vo2005}
B.~Vo, S.~Singh, and A.~Doucet, ``Sequential {M}onte {C}arlo methods for
  multi-target filtering with random finite sets.'' \emph{IEEE Transaction
  Aerospace and Electronics}, vol.~41, no.~4, pp. 1224--1245, October 2005.

\bibitem{Uney2010}
M.~\"{U}ney, S.~Julier, D.~Clark, and B.~Risti\'{c}, ``Monte carlo realisation
  of a distributed multi--object fusion algorithm.''\hskip 1em plus 0.5em minus
  0.4em\relax SSPD 2010, September 2010.

\bibitem{Uney2011}
M.~\"{U}ney, D.~Clark, and S.~Julier, ``{Information Measures in Distributed
  Multitarget Tracking},'' in \emph{Information Fusion (FUSION), 2011
  Proceedings of the 14th International Conference on}, july 2011, pp. 1 --8.

\bibitem{Uney2013}
M.~{\"U}ney, D.~E. Clark, and S.~Julier, ``Distributed fusion of {PHD} filters
  via exponential mixture densities,'' \emph{Selected Topics in Signal
  Processing, IEEE Journal of}, vol.~7, no.~3, pp. 521 -- 531, June 2013.

\bibitem{Barr2013}
J.~Barr, M.~Uney, D.~Clark, D.~Miller, M.~Porter, A.~Gning, and S.~Julier, ``A
  multi-sensor inference and data fusion method for tracking small,
  manoeuvrable maritime craft in cluttered regions,'' in \emph{Proceedings of
  the 3rd IMA on Mathematics in Defence}, 10 2013.

\bibitem{Olfati-Saber2007}
R.~Olfati-Saber, J.~A. Fax, and R.~M. Murray, ``Consensus and cooperation in
  networked multi-agent systems,'' \emph{Proceedings of the IEEE}, vol.~95,
  no.~1, pp. 215--233, Jan 2007.

\bibitem{Battistelli2013}
G.~Battistelli, L.~Chisci, C.~Fantacci, A.~Farina, and A.~Graziano, ``Consensus
  {CPHD} filter for distributed multitarget tracking,'' \emph{IEEE Journal of
  Selected Topics in Signal Processing}, vol.~7, no.~3, pp. 508--520, June
  2013.

\bibitem{Wang2017}
B.~Wang, W.~Yi, R.~Hoseinnezhad, S.~Li, L.~Kong, and X.~Yang, ``Distributed
  fusion with multi-{B}ernoulli filter based on generalized covariance
  intersection,'' \emph{IEEE Transactions on Signal Processing}, vol.~65,
  no.~1, pp. 242--255, Jan 2017.

\bibitem{Jiang2016}
M.~Jiang, W.~Yi, R.~Hoseinnezhad, and L.~Kong, ``Distributed multi-sensor
  fusion using generalized multi-bernoulli densities,'' in \emph{19th
  International Conference on Information Fusion (FUSION)}, July 2016, pp.
  1332--1339.

\bibitem{Yi2017}
W.~Yi, M.~Jiang, R.~Hoseinnezhad, and B.~Wang,
  ``\BIBforeignlanguage{English}{Distributed multi-sensor fusion using
  generalised multi-bernoulli densities},''
  \emph{\BIBforeignlanguage{English}{IET Radar, Sonar \& Navigation}}, vol.~11,
  pp. 434--443(9), March 2017.

\bibitem{Guldog2014}
M.~B. Guldogan, ``Consensus bernoulli filter for distributed detection and
  tracking using multi-static doppler shifts,'' \emph{IEEE Signal Processing
  Letters}, vol.~21, no.~6, pp. 672--676, June 2014.

\bibitem{Battistelli2015}
G.~Battistelli, L.~Chisci, C.~Fantacci, A.~Farina, and B.~N. Vo, ``Average
  {K}ullback-{L}eibler divergence for random finite sets,'' in \emph{2015 18th
  International Conference on Information Fusion (Fusion)}, July 2015, pp.
  1359--1366.

\bibitem{Li2018}
S.~Li, W.~Yi, R.~Hoseinnezhad, G.~Battistelli, B.~Wang, and L.~Kong, ``Robust
  distributed fusion with labeled random finite sets,'' \emph{IEEE Transactions
  on Signal Processing}, vol.~66, no.~2, pp. 278--293, Jan 2018.

\bibitem{phdthesis:dabak92}
A.~Dabak, ``{A Geometry for Detection Theory},'' Ph.D. dissertation, Rice
  University, Houston, TX, USA, 1992.

\bibitem{Julier2006}
S.~J. Julier, T.~Bailey, and J.~K. Uhlmann, ``Using exponential mixture models
  for suboptimal distributed data fusion,'' in \emph{Proc. of the 2006 IEEE
  Nonlinear Stat. Signal Proc. Workshop}.\hskip 1em plus 0.5em minus
  0.4em\relax Cambridge, UK: NSSPW'06, September 2006, pp. 160--163.

\bibitem{DaleyVere-Jones}
D.~Daley and D.~Vere-Jones, \emph{An introduction to the theory of point
  processes}.\hskip 1em plus 0.5em minus 0.4em\relax Springer, 1988.

\bibitem{Gunay2016}
M.~Gunay, U.~Orguner, and M.~Demirekler, ``Chernoff fusion of {G}aussian
  mixtures based on sigma-point approximation,'' \emph{IEEE Transactions on
  Aerospace and Electronic Systems}, vol.~52, no.~6, pp. 2732--2746, December
  2016.

\bibitem{Csiszar2004}
\BIBentryALTinterwordspacing
I.~Csiszar and P.~Shields, ``Information theory and statistics: A tutorial,''
  \emph{Foundations and Trends in Communications and Information Theory},
  vol.~1, no.~4, pp. 417--528, 2004. [Online]. Available:
  \url{http://dx.doi.org/10.1561/0100000004}
\BIBentrySTDinterwordspacing

\bibitem{Csiszar2003}
I.~Csiszar and F.~Matus, ``Information projections revisited,'' \emph{IEEE
  Transactions on Information Theory}, vol.~49, no.~6, pp. 1474--1490, June
  2003.

\bibitem{inproceedings:renyi60}
A.~R\'enyi, ``{On Measures of Entropy and Information},'' in \emph{Proceedings
  of the 4th Berkeley Symposium on Mathematics, Statistics and Probability},
  vol.~1, 1960, pp. 547--561.

\bibitem{chernoff1952}
\BIBentryALTinterwordspacing
H.~Chernoff, ``A measure of asymptotic efficiency for tests of a hypothesis
  based on the sum of observations,'' \emph{Ann. Math. Statist.}, vol.~23,
  no.~4, pp. 493--507, 12 1952. [Online]. Available:
  \url{https://doi.org/10.1214/aoms/1177729330}
\BIBentrySTDinterwordspacing

\bibitem{Julier2006b}
S.~J. Julier, ``An empirical study into the use of {C}hernoff information for
  robust, distributed fusion of {G}aussian mixture models,'' in
  \emph{Proceedings of the IEEE ICIF}, 2006.

\bibitem{Chen2002}
L.~Chen, P.~O. Arambel, and R.~K. Mehra, ``Fusion under unknown correlation -
  covariance intersection as a special case,'' in \emph{Proceedings of the
  Fifth International Conference on Information Fusion. FUSION 2002. (IEEE
  Cat.No.02EX5997)}, vol.~2, July 2002, pp. 905--912 vol.2.

\bibitem{Reinhardt2015}
M.~Reinhardt, B.~Noack, P.~O. Arambel, and U.~D. Hanebeck, ``Minimum covariance
  bounds for the fusion under unknown correlations,'' \emph{IEEE Signal
  Processing Letters}, vol.~22, no.~9, pp. 1210--1214, Sept 2015.

\bibitem{Nielsen2013}
F.~Nielsen, ``An information-geometric characterization of {C}hernoff
  information,'' \emph{IEEE Signal Processing Letters}, vol.~20, no.~3, pp.
  269--272, March 2013.

\bibitem{Chang2010}
K.~Chang, C.~Y. Chong, and S.~Mori, ``Analytical and computational evaluation
  of scalable distributed fusion algorithms,'' \emph{IEEE Transactions on
  Aerospace and Electronic Systems}, vol.~46, no.~4, pp. 2022--2034, Oct 2010.

\bibitem{Vo2013}
B.~T. Vo and B.~N. Vo, ``Labeled random finite sets and multi-object conjugate
  priors,'' \emph{IEEE Transactions on Signal Processing}, vol.~61, no.~13, pp.
  3460--3475, July 2013.

\bibitem{Williams2015}
J.~L. Williams, ``Marginal multi-{B}ernoulli filters: {RFS} derivation of
  {MHT}, {JIPDA}, and association-based {M}e{MB}er,'' \emph{IEEE Transactions
  on Aerospace and Electronic Systems}, vol.~51, no.~3, pp. 1664--1687, July
  2015.

\bibitem{Baccelli2016}
F.~Baccelli and J.~O. Woo, ``On the entropy and mutual information of point
  processes,'' in \emph{2016 IEEE International Symposium on Information Theory
  (ISIT)}, July 2016, pp. 695--699.

\bibitem{Hardy1934}
G.~H. Hardy, J.~E. Littlewood, and G.~P{\'o}lya, \emph{Inequalities}.\hskip 1em
  plus 0.5em minus 0.4em\relax Cambridge University Press, 1934.

\bibitem{Mahler2007a}
R.~Mahler, ``{PHD} filters of higher order in target number,'' \emph{IEEE
  Transactions on Aerospace and Electronic Sysmtes}, vol.~43, no.~4, pp.
  1523--1543, October 2007.

\bibitem{Bazaraa1993}
M.~S. Bazaraa, H.~D. Sherali, and C.~Shetty, \emph{Nonlinear Programming},
  2nd~ed.\hskip 1em plus 0.5em minus 0.4em\relax John Wiley \& Sons, Inc.,
  1993.

\bibitem{Bubeck2015}
\BIBentryALTinterwordspacing
S.~Bubeck, ``Convex optimization: Algorithms and complexity,''
  \emph{Foundations and Trends in Machine Learning}, vol.~8, no. 3-4, pp.
  231--357, 2015. [Online]. Available:
  \url{http://dx.doi.org/10.1561/2200000050}
\BIBentrySTDinterwordspacing

\bibitem{Nielsen2011a}
\BIBentryALTinterwordspacing
F.~Nielsen, ``{C}hernoff information of exponential families,'' \emph{CoRR},
  vol. abs/1102.2684, 2011. [Online]. Available:
  \url{http://arxiv.org/abs/1102.2684}
\BIBentrySTDinterwordspacing

\bibitem{Rasmussen2006}
C.~E. Rasmussen and C.~K.~I. Williams, \emph{Gaussian Processes for Machine
  Learning}.\hskip 1em plus 0.5em minus 0.4em\relax MIT Press, 2006.

\bibitem{Robert2004}
C.~P. Robert and G.~Casella, \emph{{M}onte {C}arlo Statistical Methods},
  2nd~ed.\hskip 1em plus 0.5em minus 0.4em\relax Springer, 2004.

\bibitem{Bazaraa2006}
M.~S. Bazaraa, H.~D. Sherali, and C.~M. Shetty, \emph{Nonlinear programming:
  Theory and algorithms}, 3rd~ed.\hskip 1em plus 0.5em minus 0.4em\relax John
  Wiley \& Sons, 2006.

\end{thebibliography}

\vfill

\end{document}